\documentclass[10pt, doublecolumn]{IEEEtran}
\usepackage{cite}
\usepackage{indentfirst}
\usepackage{graphicx}
\usepackage{changepage}
\usepackage{stfloats}
\usepackage{amsfonts,amssymb}
\usepackage{bm}
\usepackage{algorithm}
\usepackage{algorithmic}
\usepackage{amsmath}
\usepackage{amsmath,amssymb,amsfonts}
\usepackage{times}
\usepackage{url}
\usepackage{cite}	
\usepackage{textcomp}
\usepackage{color}
\usepackage{setspace}
\usepackage{enumitem}
\usepackage{float}
\usepackage{diagbox}
\usepackage[T1]{fontenc}
\usepackage[utf8]{inputenc}
\usepackage{authblk}
\usepackage{threeparttable}
\usepackage{booktabs}
\usepackage{footnote}
\usepackage{graphicx}
\usepackage{pifont}
\usepackage{subfigure}
\usepackage{mathrsfs}
\usepackage{bbm}
\usepackage{dsfont}
\usepackage{colortbl}
\usepackage{multirow}

\usepackage{makecell}
\usepackage[dvipsnames,table]{xcolor}

\definecolor{HeaderBlue}{HTML}{2C3E50}   
\definecolor{RowGray}{HTML}{F2F4F4}     
\definecolor{RowWhite}{HTML}{FFFFFF}    

\setcellgapes{3pt}
\makegapedcells

\allowdisplaybreaks[4]

\newenvironment{proof}{{\indent  \indent \it Proof:}}{\hfill $\blacksquare$}

\newtheorem{remark}{\bf Remark}
\newtheorem{Pro}{\bf Proposition}
\newtheorem{theorem}{\bf Theorem}

\begin{document}
\title{ISAC Network Planning: Sensing Coverage Analysis and 3-D BS Deployment Optimization}

\author{
	Kaitao Meng, \textit{Member, IEEE}, Kawon Han, \textit{Member, IEEE}, Christos Masouros, \textit{Fellow, IEEE}, and Lajos Hanzo, \textit{Life Fellow, IEEE} 
	\thanks{Kaitao Meng is with the Department of Electrical and Electronic Engineering, University of Manchester, Manchester, UK (email: kaitao.meng@manchester.ac.uk). Kawon Han is with the Department of Electrical Engineering, Ulsan National Institute of Science and Technology (UNIST), Ulsan, South Korea (email: kawon.han@unist.ac.kr). Christos Masouros is with the Department of Electronic and Electrical Engineering, University College London, London, UK (email: c.masouros@ucl.ac.uk). Lajos Hanzo is with School of Electronics and Computer Science, University of Southampton, SO17 1BJ Southampton, UK (email: lh@ecs.soton.ac.uk) }
	\thanks{The financial support of the following Engineering and Physical Sciences Research Council (EPSRC) projects is gratefully acknowledged: Platform for Driving Ultimate Connectivity (TITAN) (EP/X04047X/1; EP/Y037243/1); Robust and Reliable Quantum Computing (RoaRQ, EP/W032635/1); PerCom (EP/X012301/1); India-UK Intelligent Spectrum Innovation ICON UKRI-1859.}
}

\maketitle


\begin{abstract}
Integrated sensing and communication (ISAC) networks strive to deliver both high‐precision target localization and high‐throughput data services across the entire coverage area. In this work, we examine the fundamental trade-off between sensing and communication from the perspective of base station (BS) deployment. Furthermore, we conceive a design that simultaneously maximizes the target localization coverage, while guaranteeing the desired communication performance. In contrast to existing schemes optimized for a single target, an effective network-level approach has to ensure consistent localization accuracy throughout the entire service area. While employing time-of-flight (ToF) based localization, we first analyze the deployment problem from a localization-performance coverage perspective, aiming for minimizing the {\textit{area}} Cramér-Rao Lower Bound (A-CRLB) to ensure uniformly high positioning accuracy across the service area. We prove that for a fixed number of BSs, uniformly scaling the service area by a factor $\kappa$ increases the optimal A-CRLB in proportion to $\kappa^{2 \beta}$, where $\beta$ is the BS‐to‐target pathloss exponent. Based on this, we derive an approximate scaling law that links the achievable A-CRLB across the area of interest to the dimensionality of the sensing area. We also show that cooperative BSs extend the coverage but yield marginal A-CRLB improvement as the dimensionality of the sensing area grows. By exploiting the invariance properties discovered with respect to the displacement, rotation, and symmetric projection deformation, we derive a deployment-invariant structure for conceiving a low-complexity framework for ISAC network deployment. 
We then formulate the joint sensing-communication optimization problem and present a Majorization-Minimization algorithm for designing high-quality deployment solutions. Extensive simulations demonstrate that our framework significantly enhances sensing coverage, while maintaining the desired communication throughput.
\end{abstract}   



\section{Introduction}
High-quality  localization and data transmission are fundamental for supporting sophisticated applications, including autonomous driving and advanced augmented reality \cite{lu2021real, zhang2022artificial}. However, using separate wireless \textit{localization} networks and wireless \textit{communication} networks can lead to increased interference between the sensing and communication (S\&C) subsystems. Additionally, the rapid growth of wireless data traffic and the increasing scarcity of spectrum have inspired substantial research interests in integrated sensing and communication (ISAC) technologies \cite{Zhang2021OverviewSignal, Liu2022SurveyFundamental, Meng2023SensingAssisted, 10473676, 10216343}, to leverage a shared infrastructure and common waveforms, enabling simultaneous data transmission and echo collection for localization. Based on this unified framework, ISAC can significantly improve spectrum utilization and enhance both cost-effectiveness and energy efficiency \cite{Cui2021Integrating}. Notably, the International Telecommunication Union (ITU) has identified ISAC as one of the six key usage scenarios for the forthcoming  sixth-generation (6G) networks. To date, research efforts have primarily focused on enhancing both the sensing and communication performance within single-cell ISAC scenarios through flexible design strategies at individual base stations (BSs) \cite{Ouyang2022Performance, Liu2022Integrated, Meng2024UAV, Liu2023DistributedUnsupervised, Hua20243DMultiTarget, valiulahi2023net}. Nevertheless, substantial improvements may be attained in overall ISAC system performance at the cellular level by exploring network-level frameworks \cite{Meng2024CooperativeISACMag}. In particular, multi-cell cooperative strategies designed for joint S\&C \cite{meng2024integrated, han2025signaling} emerge as a promising hitherto under-explored research direction.

Recent studies have begun exploring the new degrees of freedom (DoF) in managing the trade-offs between communication and sensing at the network level \cite{Meng2024CooperativeTWC, han2025signaling, meng2023network, chen2022enhancing, meng2024network}. In particular, the authors of \cite{Meng2024CooperativeISACMag} and \cite{han2025network} present a systematic study of the emerging network-level performance metrics and DoF, laying the foundation for understanding and optimizing cooperative ISAC systems. As an example, \cite{Meng2024CooperativeTWC} investigates how to optimize the number of cooperating BSs to maximize the attainable networked sensing and networked communication performance under specific backhaul constraints. Coordinated precoding can repurpose inter-cell interference as useful spatial degrees of freedom, substantially improving sensing performance.
In addition, \cite{chen2022enhancing} proposed a system-level beam alignment scheme that leverages synchronization signal block and time-frequency pattern design for dramatically reducing the beam misalignment probability in THz/mmWave networks. Most recently, \cite{meng2024network} investigated antenna topology optimization in ISAC networks with randomly distributed targets, users, and BSs, introducing a new DoF in the S\&C tradeoff and showing that the optimal topology may lie between centralized massive MIMO and distributed cell free topologies, depending on the localization method and path loss exponent.
Overall, the above works mainly focus on waveform design and resource allocation, while the problem of optimizing the network deployment to enhance cooperation efficiency remains largely unexplored. 

Traditional cellular communication networks have historically been deployed to maximize communication performance, optimizing coverage probability and data throughput \cite{xu2014cooperative}. However, as future networks shift toward integrated functionalities, the dual demands of reliable data transmission and high-accuracy environmental or target sensing become increasingly critical \cite{lu2024integrated}. While communication systems are typically interference-limited, the opportunities for multi-static sensing mean that sensing systems are not interference limited, which creates a fundamental paradigm shift for the ISAC network deployment. On the other hand, inadequate BS placement can induce severe sensing blind spots, undermining seamless localization in mission-critical applications, such as autonomous driving and dynamic target tracking in smart cities. To date, no systematic study has addressed how to (i) reconfigure existing BS layouts or (ii) augment them with additional sites for jointly optimizing sensing fidelity and communication quality. Filling this knowledge gap is essential to fully exploit the network-level DoF identified in prior research \cite{Meng2024CooperativeTWC, han2025signaling, meng2023network, chen2022enhancing}, ensuring that future deployments meet the dual imperatives of robust data transmission and precise environmental awareness.

Cooperative localization using multiple BSs has recently gained attention due to its ability to achieve high localization precision at low power dissipation. In such systems, the deployment of anchor nodes, e.g., BSs, plays a critical role: an optimal deployment improves localization accuracy and also enhances the coverage probability \cite{akbarzadeh2012probabilistic}, which is essential in dynamic environments. By contrast, classical sensor placement strategies typically focus on single-target scenarios. For example, in \cite{Sadeghi2021Target}, the authors analyze the localization accuracy of a single-point target and introduce metrics such as equivalent single-radar gain, coherence gain, and geometry gain to guide the placement of transmit and receive antennas in widely separated MIMO radar systems.
Moreover, the authors of \cite{Yang2023Deployment, Jing2024ISACSky} proposed a new approach using dual-functional unmanned aerial vehicles for enhancing both the communication and localization performance in emergency scenarios by deploying them optimally to meet the ground users' needs. This allows addressing challenges such as cardinality minimization and encountering non-convex localization metrics. While these approaches optimize performance for a single target, the deployment of ISAC services in the cellular domain introduces the broader requirement of achieving robust network-wide sensing coverage. 
The influence of anchor placement on achieving uniform coverage and the corresponding scaling laws for coverage quality have not yet been rigorously established. This knowledge gap highlights the need for further investigations into deployment optimization for attaining enhanced network-wide S\&C performance.

Extending deployment strategies from isolated, single-target sensor placements to full area-wide S\&C coverage requires incorporating spatially varying path loss and geometric gain for every target location, which directly extends and reshapes the design space. Moreover, evaluating sensing performance via metrics such as the CRLB introduces intricate couplings among BS positions and network parameters, rendering analytical tractability elusive and necessitating a complete redesign of optimization algorithms \cite{shi2022device}. The resultant interdependence of BS locations not only amplifies the computational burden of determining an optimal configuration but also invalidates conventional distance‐based cooperation criteria. In cooperative multi‐BS localization, the nuanced interactions among BSs  directly influence both communication reliability and localization precision. Importantly, the scaling law that links the number of BSs to localization accuracy across an entire area of interest remains unexplored, because existing studies have been confined to the single-target scenario. Consequently, achieving robust and network‐wide ISAC performance requires novel deployment and cooperation frameworks that explicitly account for spatial diversity, parameter coupling, and the cooperative dynamics of modern localization techniques.

To overcome these challenges, we formulate the joint deployment problem as the twinned optimization of localization and communication coverage, aiming for minimizing the target localization error across the entire target area, while simultaneously ensuring the required communication quality. By examining the structure of the CRLB, we unveil that the optimal placement exhibits displacement, rotation, and symmetric projection invariance that we detail in Section \ref{DeploymentAnalysis}. These properties allow us to collapse the high-dimensional search space into a much smaller, representative subset. Such invariance-driven reduction not only yields rigorous performance bounds, characterizing, e.g., the minimum achievable average CRLB as a function of network size and area geometry, but also admits low-complexity algorithms for practical ISAC network deployments. This framework directly addresses the trade-off between communication and sensing by decoupling the interdependencies among BSs, paving the way for scalable multi-cell ISAC deployments. Then, based on the Majorization-Minimization (MM) optimization framework, we transform the problem formulated to sequential sub-problems and effectively decouple the complex interrelations among BS positions. This decoupling not only simplifies the underlying optimization problem but also allows us to have more efficient algorithmic solutions.
The main contributions of this paper are summarized as follows:
\begin{itemize}[leftmargin=*]
	\item We propose a cooperative ISAC network architecture that tightly integrates multi-static radar sensing with coordinated multi-point (CoMP) data transmission, by optimizing BS deployment to reveal the fundamental trade-offs between sensing and communication performance. Beyond traditional distributed MIMO radar deployments for individual target, our design simultaneously guarantees area‐wide high‐throughput communications and high-precision sensing.
	\item Employing time-of-flight (ToF) localization, we first establish that the optimal BS deployment strategy for  sensing performance enhancement over an area admits three provable invariances, displacement, rotation, and reflection symmetry, thereby showing that any Euclidean transformation of a candidate layout preserves its sensing-communication performance. Moreover, we prove that for a fixed number of BSs, uniformly scaling the service area by a factor $\kappa$ increases the optimal {\textit{area}} CRLB (A-CRLB) in proportion to $\kappa^{2 \beta}$, where $\beta$ denotes the pathloss exponent between BSs and targets. Leveraging the conclusions derived, we can find an optimal solution for a new area  based on the foundational optimal BS deployment strategy through displacement, rotations, symmetric projection, and area scaling.
	\item We derive an asymptotic scaling law for the minimum A-CRLB $\frac{1}{N^{\,2 \beta/d}}$, where \(N\) denotes  the number of cooperating ISAC nodes and \(d\in\{1,2,3\}\) represents  the spatial dimensionality of the sensing area. This law implies that, for a given increase in $N$, the reduction in sensing error decays more slowly as $d$ grows. Consequently, the incremental gain in coverage afforded by each additional node diminishes in higher-dimensional monitoring scenarios, underscoring an inherent trade-off between achievable sensing accuracy and the spatial extent of the surveillance area in the design of practical ISAC networks.
	\item We reformulate the general BS deployment optimization for joint sensing-communication design within the MM framework, converting it into a sequence of tractable subproblems. By embedding a trust‐region solver in each MM iteration, the resultant algorithm reliably converges to high‐quality stationary solutions, while maintaining only modest per‐iteration computational complexity. Extensive simulations demonstrate that the proposed deployment strategy is capable of enhancing the sensing‐coverage probability by at least 50\%, while maintaining the desired communication throughput.
\end{itemize}

Notation: Lower-case letters in bold font will denote deterministic vectors. For instance, $X$ and ${\bf{X}}$ denote a one-dimensional (scalar) random variable and a random vector (containing more than one element), respectively. Similarly, $x$ and ${\bf{x}}$ denote scalar and deterministic vectorial values, respectively. ${\rm{E}}_{x}[\cdot]$ represents statistical expectation over the distribution of $x$, and $[\cdot]$ represents a variable set.  Let ${\bf{1}}^T \triangleq [ 1, 1 , \dots , 1 ] \in \mathbb{R}^{1 \times N}$ and $\mathbf e_n^T \triangleq [\,0,\ldots,0,\underbrace{1}_{n\text{th}},0,\ldots,0\,]\in\mathbb{R}^{1\times N}$.
\section{System Model}

\begin{figure}[t]
	\centering
	\includegraphics[width=7.4cm]{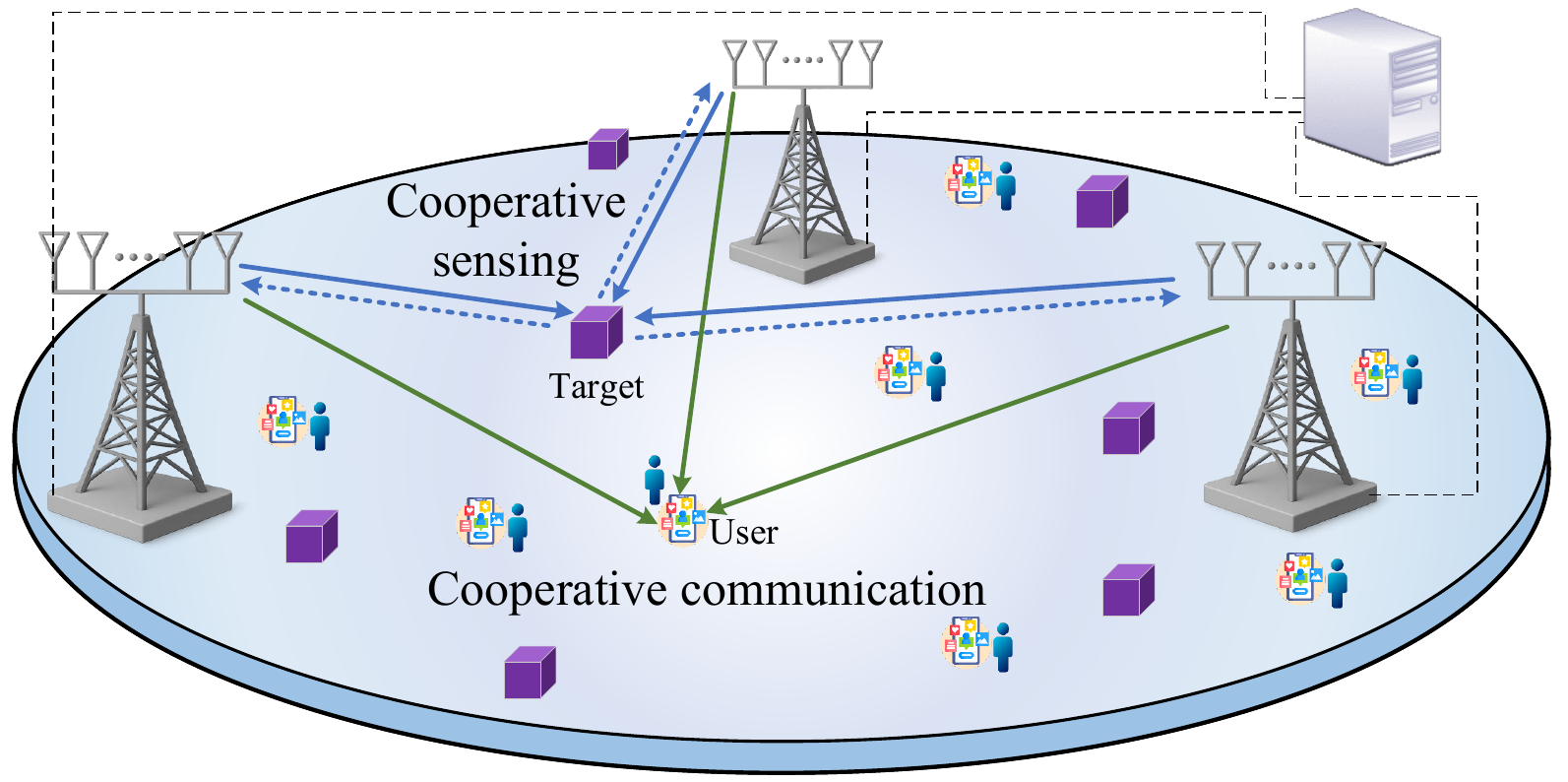}
	\vspace{0mm}
	\caption{Illustration of cooperative ISAC networks.}
	\label{figure1}
\end{figure}

\subsection{Network Model}
As shown in Fig.~\ref{figure1}, BSs within a given area form a cooperative cluster for both communication and sensing. For communication, these BSs can adopt non-coherent joint CoMP transmission, transmitting identical data streams without requiring strict phase synchronization, thereby enhancing received power, while maintaining low coordination overhead \cite{Meng2024CooperativeTWC}. For sensing, the BSs collaborate as a distributed multi-static MIMO radar system, using code-division multiplexing (CDM) to ensure orthogonality among transmitted waveforms and achieve accurate localization under non-coherent signal processing \cite{ropitault2024ieee, nitsche2014ieee}.\footnote{By avoiding phase-level synchronization, our cooperative ISAC strategy strikes a practical balance between system complexity and the performance of ISAC networks.}

\begin{table}[t]
	\small
	\centering
	\begin{threeparttable}
		\caption{Scaling Laws for Different Localization Scenarios}
		\label{tab:scaling-laws}
		\begin{tabular}{l l l l}
			\toprule
			\textbf{Target Type} &
			\textbf{Scenario} &
			\makecell{\textbf{Asymptotic} \\ \textbf{Scaling Law}} &
			\textbf{Ref.} \\
			\midrule
			\multirow{2}{*}[-2ex]{\makecell[l]{Individual\\target CRLB}}
			& Random deployment
			& $\displaystyle\frac{1}{\ln^2 N}$
			& \cite{Meng2024CooperativeTWC} \\
			\cmidrule(l){2-4}
			& Optimal deployment
			& $\displaystyle\frac{1}{N^2}$
			& \cite{Sadeghi2021Target} \\
			\cmidrule(l){1-4}
			\makecell[l]{d-dimensional \\{\textit{area}} CRLB}
			& Optimal deployment
			& $\displaystyle\frac{1}{N^{2 \beta/d}}$  $\star$
			& \textbf{Proposed} \\
			\bottomrule
		\end{tabular}
		\begin{tablenotes}
			\footnotesize
			\item [$\star$] The scaling law characterizes how the CRLB decays as the number of nodes \(N\) grows, serving as an upper bound for the average decreasing trends of CRLB with increasing \(N\) derived in this work, where $d$ denotes the dimension of the sensing area (typically $d=1,2,3$).
		\end{tablenotes}
	\end{threeparttable}
\end{table}

In this study, we explore the optimal BS deployment strategy for cooperative ISAC networks, to reveal that the deployment of BSs brings a new DoF to balance the sensing and communication performance. In this cooperative service area, each BS designs the transmit precoding for sending the information signal $s^c$ to the single-antenna communication user served, together with a dedicated radar signal $s^{s}_n$ for the detected target. Then, the communication user and the sensing target can be collaboratively served by $N$ BSs, indexed by $n$, and each BS has \(M_{\rm{t}}\) transmit antennas and \(M_{\rm{r}}\) receive antennas. Let us assume that the transmitted radar signals $\{s_n^s\}_{n=1}^N$ of the BSs in the cooperative sensing cluster are approximately orthogonal for any time delay of interest by applying CDM only to the sensing signal \(s^s_n\). Accordingly, we spread it as $\tilde s^s_n = \mathbf c_n\,s^s_n$, where $\mathbf{c}_n^H \mathbf{c}_m=\delta_{nm}$, $\|\mathbf{c}_n\|^2=1$.

For notational simplicity, we continue to use \(s^s_n\) to represent the spread signal \(\mathbf{c}_n s^s_n\) in the remainder of the discussion. It is assumed that $\mathrm{E}[s^s_n (s^c_n)^H] = 0$, which is consistent with the assumptions in \cite{Liu2020JointTransmit, Hua2023Optimal, li2008mimo}.  Upon letting ${\mathbf{s}_n=\left[s^s_n, s^c_n\right]^T}$, we have $\mathrm{E}\left[\mathbf{s}_n \mathbf{s}_n^H\right]=\mathbf{I}_2$. To facilitate our analysis, when a BS is connected to multiple users and targets, we assign them to orthogonal time or frequency resource blocks so that, within a certain resource block, each BS serves only one user and one target. For the block designated to a specific target or region of interest, all BSs are scheduled to serve that target simultaneously, enabling cooperative multistatic sensing.\footnote{Nevertheless, our analytical framework based on this assumption may be readily extended to multi-users and multi-targets association to each time/frequency resource blocks unless they are closely located within serving area.} Then, the signal transmitted by the $n$th BS is given by
\vspace{0mm}
\begin{equation}\label{TrasmitSignals}
	{\mathbf{x}}_n = \mathbf{W}_n \mathbf{s}_n =  {\bf{w}}^c_n s^c_n +   {\bf{w}}^s_n s^s_n,
	\vspace{0mm}
\end{equation}
where ${\bf{w}}^c_n$ and ${\bf{w}}^s_n \in {\mathbb {C}}^{M_{\mathrm{t}} \times 1}$ are beamforming vectors, with $\|{\bf{w}}^c_n \|^2 = p^c$ and $\|{\bf{w}}^s_n\|^2 = p^s$. Here, $p^s$ and $p^c$ respectively represent the transmit power of the sensing and communication signals, and $\mathbf{W}_n=\left[\mathbf{w}^c_n, \mathbf{w}^s_n\right] \in {\mathbb{C}}^{M_{\mathrm{t}} \times 2}$ is the transmit precoding matrix of BS $n$. The location of BS $n$ can be represented by ${\bf{b}}_n = [x_n^b, y_n^b, z_n^b]^T \in \mathbb{R}^3, n = 1, 2, \dots, N$. To eliminate intra-BS interference from sensing to communications and simplify our analysis, zero-forcing (ZF) beamforming is employed. The beamforming matrix is constructed as
\vspace{0mm}
\begin{equation}\label{TransmitBeamforming}
	{\bf W}_n
	= \tilde{\bf W}_n
	\left(\sqrt{\mathrm{diag}\left(\tilde{\bf W}_n^H\tilde{\bf W}_n\right)}\right)^{-1}
	\mathrm{diag}\big(\sqrt{p^c},\sqrt{p^s}\big).
\end{equation}
where ${\tilde {\mathbf{W}}}_n = {{\bf{H}}^H_n}{\left( {\bf{H}}_n  {\bf{H}}_n^H \right)^{-1}}$ and $\mathbf{H}_n = 
\begin{bmatrix}
	\mathbf{h}_{n,c}^H \\[0.5ex]
	\mathbf{a}^H_{M_{\mathrm{t}}}(\Omega _n)
\end{bmatrix}
\;\in\;\mathbb{C}^{2\times M_t}$. Here, \(\mathbf{h}_{n,c}^H\in\mathbb{C}^{1 \times M_{\rm{t}}}\) represents the communication channel spanning from BS \(n\) to the served user, and \(\mathbf{a}^H(\Omega _n)\) corresponds to the sensing channel impinging from BS \(n\) to the target. Here, $\Omega_n$ denotes the spatial frequency associated with angles of arrival from BS $n$ to the target point ${\bf{t}} = [x^t, y^t, z^t]^T \in \mathbb{R}^3$ along the antenna array, for an uniform linear array aligned with the $x$-axis with spacing $\lambda/2$. With azimuth-elevation angles $(\varphi_n,\theta_n)$ from BS $n$ to the target, we have $\Omega_n=\pi\,\sin\theta_n\cos\varphi_n$. The spatial frequency $\Omega_n$ is the per-element phase increment induced by the direction $(\varphi_n,\theta_n)$ projected onto the array axis, i.e., $\mathbf{a}_{M_{\mathrm{t}}}(\Omega_n)=\big[1,\ e^{\mathrm{j}\Omega_n},\ \ldots,\ e^{\mathrm{j}(M_{\mathrm{t}}-1)\Omega_n}\big]^{\top}$. 
With the aid of ZF beamforming, intra-cell communication interference is minimized since all cooperating BSs serve the same communication user, while their sensing beams avoid that user.

\subsection{Cooperative Sensing Model}
\label{CooperativeSensing}
We aim to explore the optimal BS deployment method by examining time-of-flight based ranging measurements. The location of a certain target point is denoted as ${\bf{t}} \in \mathbb R^{3} $, where ${\bf{t}} \in \mathcal{A}$, and $\mathcal{A}$ represents the entire area of interest.
The base-band equivalent of the signal reflected from target point ${\bf{t}}$ at receiver $m$ is represented as 
\vspace{0mm}
\begin{equation}\label{SensingChannel}
	\begin{aligned}
		{{\bf{y}}_{m}}(\tau) \!& =  \!\sum\nolimits_{n = 1}^N \sigma \underbrace {{\| {\bf{t}} - {\bf{b}}_m\|^{ - \frac{\beta}{2} }}{\bf{a}}_{M_{\rm{r}}}( \Omega_{m} ){\| {\bf{t}} - {\bf{b}}_n\|^{ - \frac{\beta}{2} }}{{\bf{a}}}^H_{M_{\rm{t}}}( \Omega_{n} )}_{{\text{target channel}}}  \\
		& {{\bf{W}}_n}{{\bm{s}}_n}(\tau - \! \tau_{n,m}) \! + \!\underbrace {\sum\nolimits_{n \in \Phi_I} \! \! { {{\bf{H}}_{n,m}} \mathbf{W}_n \mathbf{s}_n(\tau \! - \! \tilde{\tau}_{n,m}) }}_{{\text{inter-cluster interference}}} \! +  {\bf{n}}(\tau),
	\end{aligned}
\end{equation}
where $\beta \ge 2$ is the pathloss exponent between the serving BS and the target.\footnote{For sensing performance analysis, we consider an open area of interest for target sensing where no blockage exists. This is because the radar sensing normally focuses on the light-of-sight (LoS) targets parameters and Non-LoS (NLoS) reflections may be neglected due to additional high pathloss of double-bouncing nature of radar.} Furthermore, $\sigma$ denotes the radar cross section (RCS), $\tau _{n,m}$ is the propagation delay of the bistatic link spanning from BS $n$ to the target and then to BS $m$, while $\tilde \tau _{n,m}$ denotes the propagation delay of the direct link impinging from BS $n$ to BS $m$. Recognizing that a link-invariant RCS is an idealization in multistatic settings, we adopt an average effective $\sigma$ for tractable deployment optimization, consistent with \cite{mishra2019toward}. In (\ref{SensingChannel}), ${\bf{H}}_{n,m}$ denotes the channel from BS $n$ to BS $m$. Finally, the term ${\bf{n}}(\tau)$ is the additive complex Gaussian noise having zero mean and covariance matrix $\sigma_s^2 {\bf{I}}_{M_{\rm{r}}}$.

Assuming unbiased estimations, the CRLB serves as a benchmark for theoretical localization accuracy in terms of the mean squared error (MSE), which can be expressed as
\vspace{0mm}
\begin{equation}	
	\mathbb E\left[\|\hat{\mathbf t}-\mathbf t\|^2\right]
	\ge \mathrm{Tr}\left(\mathbf F(\mathbf t)^{-1}\right)
	\triangleq \mathrm{CRLB}(\mathbf t),
	\vspace{0mm}
\end{equation}
where $\hat{{\bf{t}}}=\left[\hat{x}^t, \hat{y}^t, \hat{z}^t\right]^T$ represents the estimated location of the target. 
By applying matched filtering and reciprocal filtering, one can estimate the ToF of the signal between targets and BSs, and hence determine the propagation distance. Specifically, from transmitter $n$ to the target and then to receiver $m$, the term \( \hat d_{nm} \) denotes the bistatic range corresponding to the path BS-$n$ $\to$ target $\to$ BS-$m$, which is given by
\vspace{0mm}
\begin{equation}
	\hat d_{nm} = d_{nm} + n^t_{nm},
	\vspace{0mm}
\end{equation}
where $d_{nm} = \|\mathbf b_n-\mathbf t\|+\|\mathbf b_m-\mathbf t\|$ denotes the true distance, the measurement noise is $ n^t_{nm} \sim \mathcal{N}\left(0, \eta_{nm}^2\right)$, $\eta_{nm}^2 = \frac{ 3 c^2 \sigma_s^2  }{8 \pi^2 G_t M_r B^2  \gamma_{nm}}$, and $\gamma_{nm}
= \bar\sigma  \|\mathbf b_n-\mathbf t\|^{-\beta}\, \|\mathbf b_m-\mathbf t\|^{-\beta}$, the constant $\bar\sigma >0$ absorbs distance-invariant factors such as the average RCS. Here, $c$ denotes the speed of light, $B^2$ represents the effective squared bandwidth, $G_t$ is the transmit beamforming gain in the direction of the target, and $\gamma_{nm}$ represents the bistatic channel power.

Then, we transform $N^2$ time-of-flight measurement links into the target location. The Jacobian of the $N^2$ bistatic range measurements, evaluated at the true target position $\bf{t}$,\footnote{In our system, while we may not have the exact target location, the estimated or predicted location of the target gleaned from previous sensing results can be utilized for performance analysis, such as in target tracking.} can be expressed as
\[\mathbf J(\mathbf t)=
\begin{bmatrix}
	\left(\frac{(\mathbf t-\mathbf b_1)^T}{\|\mathbf t-\mathbf b_1\|}
	+\frac{(\mathbf t-\mathbf b_1)^T}{\|\mathbf t-\mathbf b_1\|}\right)\\
	\left(\frac{(\mathbf t-\mathbf b_1)^T}{\|\mathbf t-\mathbf b_1\|}
	+\frac{(\mathbf t-\mathbf b_2)^T}{\|\mathbf t-\mathbf b_2\|}\right)\\
	\vdots\\
	\left(\frac{(\mathbf t-\mathbf b_n)^T}{\|\mathbf t-\mathbf b_n\|}
	+\frac{(\mathbf t-\mathbf b_m)^T}{\|\mathbf t-\mathbf b_m\|}\right)\\
	\vdots\\
	\left(\frac{(\mathbf t-\mathbf b_N)^T}{\|\mathbf t-\mathbf b_N\|}
	+\frac{(\mathbf t-\mathbf b_N)^T}{\|\mathbf t-\mathbf b_N\|}\right)
\end{bmatrix}\in\mathbb R^{N^2\times 3}.\]
Then, the Fisher information matrix (FIM) of sensing performance related to the target point is given by
\vspace{0mm}
\begin{equation}\label{FisherMatrix}
	\mathbf {F}\left({\bf{t}}\right)=\mathbf {J}^{T}\mathbf {R}^{-1}\mathbf {J},
\vspace{0mm}
\end{equation}
where $
\mathbf R^{-1}=\mathrm{diag}\left(\eta_{nm}^{-2}\right)_{(n,m)}.
$
Since the transmit beamforming gain \( G_t \) in the direction of the target is influenced by the randomness of the communication user channels in the joint waveform design (c.f. (\ref{TransmitBeamforming})), we perform an expected CRLB analysis based on the channel statistics.  
Given the locations of the ISAC BSs, the expected CRLB at a specific target point \( \mathbf{t} \) can be given by
\vspace{0mm}
\begin{equation}
	\text{CRLB} ({\bf{t}}) = {\rm{E}}_{G_t} \left[\text{Tr}\left(\mathbf{F}\left({\bf{t}}\right)^{-1}\right)\right] .
\vspace{0mm}
\end{equation}
To evaluate the overall positioning performance across the entire area of interest, denoted by \( \mathcal{A} \), we define the expected \emph{area} CRLB (A-CRLB) as the average of \( \text{CRLB}(\mathbf{t}) \) over \( \mathcal{A} \), given by
\vspace{0mm}
\begin{equation}
	\overline{\text{CRLB}} =  \frac{1}{|\mathcal{A}|} \int_{ \bf{t}  \in {\mathcal{A}}} {{\rm{CRLB}}\left( \bf{t} \right)  d {\bf{t}}} , 
	\vspace{0mm}
\end{equation}
where $|\mathcal{A}|$ denotes the size (or volume) of the sensing area $\mathcal{A}$ in a $d$-dimensional space.  As summarized in Table \ref{tab:scaling-laws}, the {individual-target CRLB} is the pointwise CRLB on localization error for a single target. By contrast, the $d$-dimensional area CRLB averages $\mathrm{CRLB}(\mathbf t)$ over a region $A\subset\mathbb{R}^d$ to assess localization performance across the entire area.

\subsection{Cooperative Communication Model}

We assume that the transmitters use non-coherent joint transmission, where the useful signals are combined by accumulating the transmit signal powers of cooperative BSs, i.e., all cooperating BSs transmit the same communication signal, with $\mathbf{s}_n =\left[s^s_n, s^c\right]^T$. Here, ${\bf{u}} = [x^u, y^u, z^u]^T \in \mathbb{R}^3$ denotes the user location. The signal received at the typical user is then given by
\vspace{0mm}
	\begin{align}
		y_{c}=& \underbrace{\sum\nolimits_{n=1}^N \|{\bf{b}}_n - {\bf{u}} \|^{ - \frac{\alpha}{2} } \mathbf{h}_{n}^H \mathbf{W}_n \mathbf{s}_n}_{\text{collaborative intended signal}} \nonumber \\
		& +  \underbrace{\sum\nolimits_{{n \in \Phi_I}}\|{\bf{b}}_n - {\bf{u}} \|^{-\frac{\alpha}{2}} \mathbf{h}_{n}^H \mathbf{W}_n \mathbf{s}_n}_{\text{inter-cluster interference}} + {n_{c}},
		\vspace{0mm}
	\end{align}
where $\alpha \ge 2$ is the pathloss exponent, $\mathbf{h}^H_{n} \sim \mathcal{C N}\left(0, \mathbf{I}_{M_{\mathrm{t}}}\right)$ is the channel vector of the link between the BS at $\mathbf{b}_n$ and the communication user, $\Phi_I$ is the interfering BS set, and $n_{c}\in \mathcal{C N}(0,\sigma^2_c)$ denotes the noise. Since our goal is to optimize BS placement within a fixed local area, and external interference  cannot be readily controlled, we simplify performance assessment to the signal-to-noise ratio (SNR) of the cooperative cluster, treating the inter-cluster interference term as noise \cite{Park2016OptimalFeedback}.
The SNR of the received signal at the user ${\bf{u}}$ can be expressed as
\vspace{0mm}
 \begin{equation}\label{SNRexpression}
	{\rm{SNR}}_c = \frac{ {\sum\nolimits_{n = 1}^N {{g_n}{{\|{\bf{b}}_n - {\bf{u}} \|}^{ - \alpha }}} }}{\tilde \sigma_c^2},
	\vspace{0mm}
\end{equation}
where  $\left|\mathbf{h}_{n}^H \mathbf{w}_n^c\right|^2$ denotes the effective desired signals' channel gain, and $\tilde \sigma_c^2$ denotes the effective noise variance, including interference.
The average data rate of users is given by 
\vspace{0mm}
\begin{equation}
	R_c({\bf{u}}) = \mathrm{E}_{g_n}[\log_2 (1+\mathrm{SNR}_c)].
	\vspace{0mm}
\end{equation}
We then define the \emph{area} communication rate over the communication user area $\mathcal B$ by $\frac{1}{{|\mathcal{B}|}} \int_{\bf{u} \in {\mathcal{B}}} {{R_c}\left( \bf{u} \right) {\rm{d}}\bf{u}}$. 
Similarly, 
$|\mathcal{B}|$ denotes the size (or volume) of the communication user area ${\mathcal{B}}$.

\subsection{Problem Formulation}

We aim for minimizing the expected localization error over the area of interest while satisfying the communication rate requirement. In practice, although the exact user and target locations are unknown a priori, this assumption enables us to characterize the worst-case service performance across the entire area, reflecting practical deployment requirements.  Therefore,
the problem can be formulated as
\begin{alignat}{2}
	\label{P1}
	(\rm{P1}): & \begin{array}{*{20}{c}}
		\mathop {\min }\limits_{\{{\bf{b}}_n\}_{n=1}^N} \quad  \frac{1}{|{\cal{A}}|}  \int_{ \bf{t}  \in {\cal{A}}} {{\rm{CRLB}}\left( \bf{t} \right)  {\rm{d}}\bf{t}} 
	\end{array} & \\ 
	\mbox{s.t.}\quad
	& \frac{1}{|{\cal{B}}|} \int_{\bf{u} \in {\cal{B}}} {{R_c}\left( \bf{u} \right) {\rm{d}}\bf{u}} \ge R^{th}, & \tag{\ref{P1}a}
\end{alignat}
where \({\bf{b}}_n\) denotes the position of the \(n\)th BS.\footnote{While our study treats static BS placement, the framework naturally extends to mobile BSs (e.g., drones) by re-optimizing time-varying locations on a receding horizon using online user and target trajectory estimates, subject to mobility and energy constraints, to minimize a time-averaged area-CRLB.}  

\begin{remark}
In contrast to communication sum rate maximization, where the sum rate can theoretically grow unbounded as the SNR increases, fairness is difficult to ensure. The system tends to allocate more resources to users having better channel conditions, potentially leaving users with poor channels underserved or even entirely ignored. However, since the CRLB is strictly positive, minimizing the area CRLB can be heavily influenced by sub-regions having high CRLB values, even if the CRLB is near zero across most of the sensing area. This characteristic makes the optimization highly sensitive to regions with poor performance, compelling BS deployment to prioritize these regions for improvement. Consequently, minimizing the area CRLB is inherently sensitive to worst-case locations. Explicitly, if a sub-region's CRLB tends to infinity, so does the objective, forcing BS deployment to eliminate worse sensing points and thus yielding a naturally balanced performance over the area of interest without an explicit upper-bound constraint of $
		\max_{\mathbf{t}_k}\mathrm{CRLB}(\mathbf{t}_k)\le\mathrm{CRLB}^{\mathrm{th}}
		$. 
\end{remark}

\section{Optimal Sensing Deployment Analysis}
\label{DeploymentAnalysis}
While the communication-only deployment has been well studied in the past decades, we have to understand the effects of BS deployment on the area sensing performance before designing the joint ISAC deployment. In this section, we first prove several important characteristics of an optimal deployment strategy, to reveal the unique scaling law of sensing coverage quality and facilitate solving the problem.

\subsection{Invariance Characteristics of Optimal Deployment}

In this section, we first characterize the optimal BS deployment from a sensing performance perspective by first ignoring any communication constraints. Specifically, the sensing-optimal deployment problem can be formulated as follows:
\vspace{0mm}
\begin{alignat}{2}
	\label{P1.1}
	(\rm{P1.1}): & \begin{array}{*{20}{c}}
		\mathop {\min }\limits_{\{{\bf{b}}_n\}_{n=1}^N} \quad  \frac{1}{|{\cal{A}}|}\int_{ \bf{t}  \in {\cal{A}}} {{\rm{CRLB}}\left( \bf{t} \right) {\rm{d}}\bf{t}} .
	\end{array} 
\end{alignat}
Building on this optimality formulation, we demonstrate that the area CRLB is independent of any displacement, rotation, and symmetric projection under optimal BS deployment, as shown in Fig. \ref{figure2}, based on which, we present a streamlined design procedure for sensing networks.

\begin{Pro}\label{TranslationInvariance}
	{\textbf{Invariance to displacement.}} Let \(\mathcal{A}\subset\mathbb{R}^d\) be the sensing area in problem \((\mathrm{P1.1})\), and let $
	\{{\bf{b}}_n^*\}_{n=1}^N$
	be an optimal BS deployment attaining the minimum $\overline{\text{CRLB}}$
	$=
	E^*$.
	For any displacement vector \(\boldsymbol\Delta\in\mathbb{R}^d\), define the shifted area
	$
	\mathcal{A}_{\boldsymbol\Delta}
	=\mathcal{A}+\boldsymbol\Delta$.
	Then,
	$
	\{{\bf{b}}_n^* + \boldsymbol\Delta\}_{n=1}^N
	$
	is an optimal solution of \((\mathrm{P1.1})\) posed over \(\mathcal{A}_{\boldsymbol\Delta}\), achieving the same objective value \(E^*\). 
\end{Pro}

\begin{proof}
	The performance metric CRLB in (\ref{FisherMatrix}) depends only on the relative positions between BSs and target points, i.e., $\{{\bf{b}}_n - \mathbf{t}\}$. As shown in Fig. \ref{figure2}(a), after displacing all BSs by \( \mathbf{\Delta} \), the new positions are  ${\bf{b}}_n' = {\bf{b}}_n + \mathbf{\Delta}$ and $\mathbf{t}' = \mathbf{t} + \mathbf{\Delta}$. For any BS \( n \) and target point $\mathbf{t}$,  we have
	\vspace{0mm}
	\begin{equation}
	{\bf{b}}_n' - \mathbf{t}' = ({\bf{b}}_n + \mathbf{\Delta}) - (\mathbf{t} + \mathbf{\Delta}) = {\bf{b}}_n - \mathbf{t}.
	\vspace{0mm}
	\end{equation}  
	Thus, relative positions remain unchanged, and the performance metric is displacement-invariant, i.e., $\mathrm{CRLB}(\mathbf{t}+\boldsymbol\Delta)\;=\;\mathrm{CRLB}(\mathbf{t})$. If \( {\bf{b}}_n^* \) is optimal for area \( \mathcal{A} \), then \( {\bf{b}}_n^* + \mathbf{\Delta} \) also achieves the same optimal metric \( E^* \). By the method of contradiction, suppose $
	\{{\bf{b}}_n^* + \boldsymbol\Delta\}_{n=1}^N
	$ were not optimal on $\mathcal A_{\boldsymbol\Delta}$, i.e., a certain deployment achieves $E'<E^*$ there; then displacing it back by $-\boldsymbol\Delta$ yields $E'<E^*$ on $\mathcal A$, contradicting the optimality of $
	\{{\bf{b}}_n^*\}_{n=1}^N$.
\end{proof}

\begin{figure}[t]
	\centering
	\includegraphics[width=8.4cm]{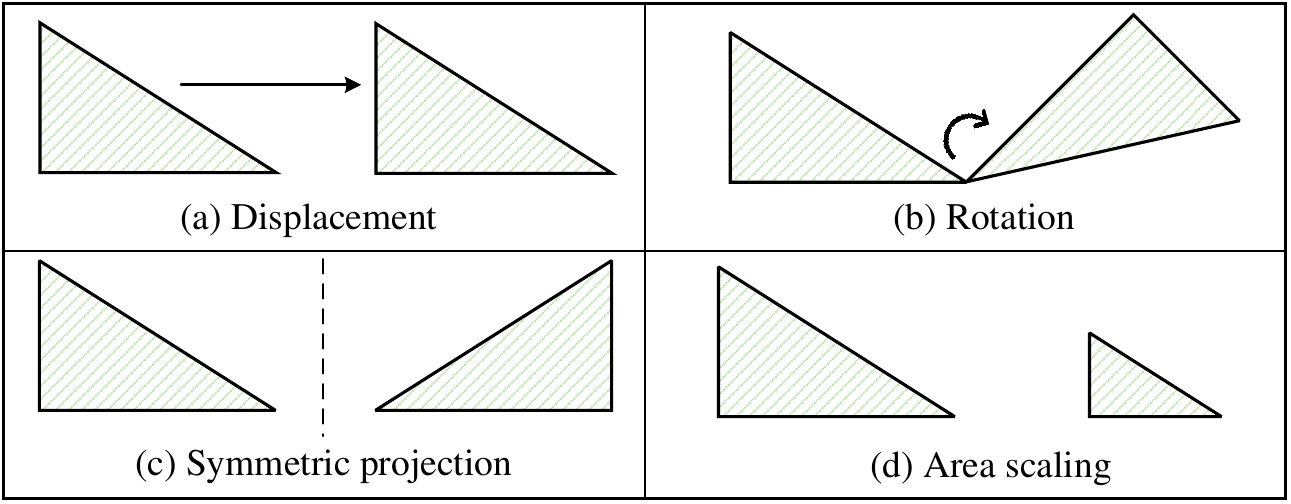}
	\vspace{0mm}
	\caption{Illustration of invariance and similarity.}
	\label{figure2}
\end{figure}

\begin{Pro}\label{RotationInvariance}
	{\textbf{Invariance to Rotation}}. Let \(\mathcal{A}\subset\mathbb{R}^d\) be the sensing area in problem \((\mathrm{P1.1})\).  Suppose
	$
	\{{\bf{b}}_n^*\}_{n=1}^N
	$
	is an optimal BS deployment for \((\mathrm{P1.1})\), attaining the minimum $\overline{\text{CRLB}}$
	$=
	E^*$.
	 Let \({\bf{R}}_{\varphi}\in\mathbb{R}^{d \times d}\) be any arbitrary rotation matrix, i.e.
	$
	{\bf{R}}_{\varphi}{\bf{R}}_{\varphi}^T = \mathbf{I}, 
	\det({\bf{R}}_{\varphi})=1,
	$
	and define the rotated regions
	$
	\mathcal{A}_{{\bf{R}}_{\varphi}} ={\bf{R}}_{\varphi} \mathcal{A}$.
	Then
	$
	\{{\bf{R}}_{\varphi}\,{\bf{b}}_n^*\}_{n=1}^N
	$
	is an optimal solution of \((\mathrm{P1.1})\) over \(\mathcal{A}_{{\bf{R}}_{\varphi}}\), achieving the same objective value \(E^*\).
\end{Pro}
\begin{proof}
	For any ${\bf{b}}, \mathbf{t}\in\mathbb{R}^d$,
		$
		\|{\bf{R}}_{\varphi}{\bf{b}} - {\bf{R}}_{\varphi}\mathbf{t}\|
		= \|{\bf{b}}-\mathbf{t}\|,
		$
		we have
		$
		\mathrm{CRLB}({\bf{R}}_{\varphi}\mathbf{t}) = \mathrm{CRLB}(\mathbf{t}).
		$
		Therefore, it follows that 
		$
		\frac{1}{|\mathcal{A}_{{\bf{R}}_{\varphi}}|}\int_{\mathcal{A}_{{\bf{R}}_{\varphi}}} \mathrm{CRLB}(\mathbf{x})\mathrm{d}\mathbf{x}
		= \frac{1}{|\mathcal{A}|} \int_{\mathcal{A}} \mathrm{CRLB}(\mathbf{t})\mathrm{d}\mathbf{t} = E^*.
		$
		If a deployment on $\mathcal{A}_{{\bf{R}}_{\varphi}}$ achieved an objective less than $E^*$, rotating it back by ${\bf{R}}_{\varphi}^T$ would yield a deployment on $\mathcal{A}$ having an objective below $E^*$, contradicting the optimality of $\{{\bf{b}}_n^*\}$. Hence $\{{\bf{b}}_n'\}$ is optimal on $\mathcal{A}_{{\bf{R}}_{\varphi}}$. Similar to the proof in Proposition \ref{TranslationInvariance}, it can be shown that $
		\{{\bf{R}}_{\varphi}\,{\bf{b}}_n^*\}_{n=1}^N
		$
		is an optimal solution of \((\mathrm{P1.1})\) over \(\mathcal{A}_{{\bf{R}}_{\varphi}}\).
\end{proof}

According to Proposition~\ref{RotationInvariance}, let \(\mathcal{R}\) be a region that is invariant under a rotation by angle \(\varphi\), i.e.,  
$
\mathcal{R} = \mathcal{R}\,{\bf{R}}_{\varphi},
$
where \({\bf{R}}_{\varphi}\) is the \(\varphi\)-rotation matrix. If \(\mathbf{b}^*\) is an optimal deployment in \(\mathcal{R}\), then for each \(i=1,\dots,k-1\) the rotated configuration
$
R_{\varphi}^i\,\mathbf{b}^*
$
is also optimal. In particular, if \(\mathcal{R}\) has \(k\)-fold rotational symmetry, there exist \(k\) distinct yet equivalent optimal solutions.

\begin{Pro}\label{MirrorSymmetricRegion}
	{\textbf{Invariance under Symmetric Projection}}. Let \(\mathcal{A}\subset\mathbb{R}^d\) be the sensing region in problem \((\mathrm{P1.1})\).  Let
	$
	\mathbf{O}_H \in \mathbb{R}^{d\times d},\;
	\mathbf{O}_H\mathbf{O}_H^T = \mathbf{I},\;\det(\mathbf{O}_H)=-1,
	$
	be the symmetric projection matrix across some hyperplane \(H\), and define its symmetric projection region
	$
	\mathcal{A}' \;=\; \mathbf{O}_H\,\mathcal{A}.
	$
	If \(\{\mathbf{b}_n^*\}_{n=1}^N\) is an optimal BS deployment on \(\mathcal{A}\) achieving
	\(\overline{\mathrm{CRLB}} = E^*\), then the reflected deployment
	\(\{\mathbf{O}_H\,\mathbf{b}_n^*\}_{n=1}^N\)
	is optimal on \(\mathcal{A}'\), achieving the same
	\(\overline{\mathrm{CRLB}} = E^*\).
\end{Pro}

\begin{proof}
	Let \(\mathbf{t}\in\mathcal{A}\) be any target point and \(\mathbf{t}' = \mathbf{O}_H\,\mathbf{t}\) belong to its symmetric projection area.  Feasibility holds, since \(\mathbf{t} \in \mathcal{A}\implies \mathbf{O}_H \mathbf{t}\ \in \mathcal{A}'\).  
	Under symmetric projection, each row of the Jacobian transforms as
	\vspace{0mm}
	\begin{equation}
		\frac{\bigl(\mathbf{O}_H\,\mathbf{b}_n^* - \mathbf{O}_H\,\mathbf{t}\bigr)^T}
		{\bigl\|\mathbf{O}_H\,\mathbf{b}_n^* - \mathbf{O}_H\,\mathbf{t}\bigr\|}
		=
		\frac{(\mathbf{b}_n^* - \mathbf{t})^T\,\mathbf{O}_H^T}
		{\|\mathbf{b}_n^* - \mathbf{t}\|},
		\vspace{0mm}
	\end{equation}
	so 
	\(\mathbf{J}(\mathbf{t}') = \mathbf{J}(\mathbf{t})\,\mathbf{O}_H^T\).  
	Since \(\|\mathbf{O}_H {\bf{x}} - \mathbf{O}_H {\bf{y}}\| = \|{\bf{x}} - {\bf{y}}\|\), the precision-noise matrix \(\mathbf{R}^{-1}\) is unchanged.  Hence the Fisher information transforms as
	\vspace{0mm}
	\begin{equation}
		\mathbf{F}(\mathbf{t}')
		= (\mathbf{O}_H\,\mathbf{J}^T)\,\mathbf{R}^{-1}\,(\mathbf{J}\,\mathbf{O}_H^T)
		= \mathbf{O}_H\,\mathbf{F}(\mathbf{t})\,\mathbf{O}_H^T,
		\vspace{0mm}
	\end{equation}
	and therefore
	\((\mathbf{F}(\mathbf{t}'))^{-1}
	= \mathbf{O}_H\,\mathbf{F}(\mathbf{t})^{-1}\,\mathbf{O}_H^T\).  
	By the cyclic property of the trace, we have
	\(\mathrm{Tr}\bigl((\mathbf{F}(\mathbf{t}'))^{-1}\bigr)
	= \mathrm{Tr}\bigl(\mathbf{F}(\mathbf{t})^{-1}\bigr)\),
	so the CRLB is invariant. Similar to the contradiction-based proof in Proposition \ref{TranslationInvariance}, it can be shown that the reflected deployment attains the same minimum $\overline{\text{CRLB}} = E^*$ and it is optimal.
\end{proof}

The foregoing analysis has identified a set of invariance properties characterizing optimal BS deployments over a fixed‐area region, which can facilitate low-complexity algorithm design, as detailed in Section \ref{LowComplixity}. In the next section, we extend this framework to variable‐area scenarios and demonstrate how those invariances yield simple scaling laws that, in turn, lead to low‐complexity algorithms for designing deployments that achieve optimal sensing coverage.

\subsection{Scaling Law Analysis}

In this section, we now rigorously examine how these optimal configurations behave under geometric transformations, deriving the scaling law vs. the number of BSs for a specific area CRLB value.

\begin{Pro}\label{ScaleInvariance}
{\textbf{Uniform Area Scaling}}. Let \(\mathcal{A}\subset\mathbb{R}^d\) be a region whose optimal BS deployment \(\{\mathbf{b}_n^*\}_{n=1}^N\) attains $\overline{\text{CRLB}} = $ \(E^*\).  If we scale the entire region by \(\kappa>0\), i.e.\ replace \(\mathcal{A}\) with \(\kappa\,\mathcal{A}=\{\kappa\,\mathbf{x}:\mathbf{x}\in\mathcal{A}\}\), then the new optimal deployment is \(\{\kappa\, \mathbf{b}_n^*\}_{n=1}^N\)  and the optimal $\overline{\text{CRLB}}$ becomes \(\kappa^{2\beta}E^*\).
\end{Pro}

\begin{proof}
	Denote any target position by \(\mathbf t\) and its scaled version by \(\mathbf t'=\kappa\,\mathbf t\), as shown in Fig.~\ref{figure2}(d).  Similarly \({\bf{b}}_n'=\kappa\,{\bf{b}}_n\). 
	For the scaled region, each row of the Jacobian matrix $\mathbf J(\mathbf t)$ becomes
	\vspace{0mm}
	\begin{equation}
		\frac{({\bf{b}}_n'-\mathbf t')^T}{\|{\bf{b}}_n'-\mathbf t'\|}
		=\frac{(\kappa{\bf{b}}_n-\kappa\mathbf t)^T}{\kappa\|{\bf{b}}_n-\mathbf t\|}
		=\frac{({\bf{b}}_n-\mathbf t)^T}{\|{\bf{b}}_n-\mathbf t\|},
		\vspace{0mm}
	\end{equation}
	so the scaled Jacobian is \(\mathbf J(\mathbf t')=\mathbf J(\mathbf t)\).
	Ignoring fixed multiplicative constants, the original inverse‐covariance \(\mathbf R^{-1}\) is diagonal with entries of the form  \(1/(\|{\bf{b}}_n-\mathbf t\|^{\beta}\|{\bf{b}}_m-\mathbf t\|^{\beta})\) for the bi-static link from BS $n$ to a target and then to BS $m$. After scaling, we have:
	\vspace{0mm}
	\begin{equation}
		\|{\bf{b}}_n'-\mathbf t'\|=\kappa\,\|{\bf{b}}_n-\mathbf t\|
		\;\Longrightarrow\;
		\mathbf R'^{-1}=\kappa^{- 2 \beta}\,\mathbf R^{-1}.
		\vspace{0mm}
	\end{equation}
	Then, the FIM of a target point can be formulated as:
	\vspace{0mm}
	\begin{equation}
		\mathbf F(\mathbf t')
		=\mathbf J^T\bigl(\kappa^{- 2 \beta}\mathbf R^{-1}\bigr)\mathbf J
		=\kappa^{- 2 \beta}\,\mathbf F(\mathbf t).
		\vspace{0mm}
	\end{equation}
	The CRLB is \(\mathrm{Tr}(\mathbf F^{-1})\).  Thus, we have:
	\vspace{0mm}
	\begin{equation}
			\overline{\mathrm{CRLB}}\bigl(\mathcal A',\{\mathbf b_n^*\}\bigr) = \kappa^{2 \beta} 
			\overline{\mathrm{CRLB}}\bigl(\mathcal A,\{\mathbf b_n^*\}\bigr).
			\vspace{0mm}
	\end{equation}
	Assume for the method of contradiction that $\{\kappa\,\mathbf b_n^*\}$ is not optimal on $\kappa\,\mathcal A$, so there exists another deployment achieving $\overline{\mathrm{CRLB}}\bigl(\mathcal A',\{\mathbf b_n^*\}\bigr) < \kappa^{{2 \beta} }E^*$.  Rescaling that deployment by $1/\kappa$ yields a deployment on $\mathcal A$ with $\overline{\mathrm{CRLB}}<E^*$, contradicting the optimality of $\{\mathbf b_n^*\}$.  Therefore $\{\kappa\,\mathbf b_n^*\}$ must be optimal and the $\overline{\mathrm{CRLB}}$ is exactly $\kappa^{2 \beta} E^*$.
	Hence scaling the region by \(\kappa\) maps the optimal deployment \({\bf{b}}^*\to\kappa\,{\bf{b}}^*\) and scales the CRLB by \(\kappa^{2 \beta} \).  
\end{proof}

By exploiting invariance to rigid motions (global translations, rotations, and reflections) together with the uniform-area scaling in Proposition~4, an optimal BS deployment computed for a canonical region can transfer directly to any translated/rotated/reflected/scaled instance of that region. Based on the conclusions of Propositions 1-4 above, we can further characterize the optimal deployment features for large-scale regions. In particular, for extensive areas the optimal deployment exhibits clear regularity (defined here as the presence of repetitive, periodic, or structured spatial arrangements), thereby significantly simplifying the design process.

\begin{remark}\label{RegularRegion}
	{\textbf{Area Periodicity}}. Although a rigorous proof in the general case is still open, it is intuitively reasonable that the optimal deployment $\mathbf{b}^*$ is $\Lambda$-periodic, whenever the sensing region is sufficiently large, effectively a tiling of identical subregions, and BSs are densely deployed. In other words, for any translation vector $v\in\Lambda$, the optimal deployment follows the invariance to displacement principle as described in Proposition \ref{TranslationInvariance}. 
	The key intuition is that the total cost decomposes into a sum of identical local costs over each periodic sub-region, each depending mainly on the pattern within that sub-region and invariant under $\Lambda$-translations. If two sub-regions had different deployment patterns, one with higher cost, one could copy the lower-cost pattern into the other sub-regions to reduce the overall CRLB, contradicting the assumed optimality of $\mathbf{b}^*$. Therefore, the optimal deployment must repeat the same pattern across all sub-regions, yielding $\Lambda$-periodicity.
\end{remark}

Moreover, upon leveraging Proposition \ref{ScaleInvariance} and the analysis in Remark \ref{RegularRegion}, we can derive an upper bound on the scaling law of the area CRLB as follows, which can serve as a guideline for ISAC network deployment.

\begin{theorem}\label{BSCountScaling}
	Let $n_0$ be a fixed constant number of BSs optimally deployed in a sufficiently large region  $\mathcal{A}\subset\mathbb{R}^d$, achieving optimal $\overline{\rm{CRLB}}  = E^*(n_0,\mathcal{A})$.  If the BS count is increased by a factor of $N$, i.e., the total BS number becomes $N \times n_0$, then the  optimal $\overline{\rm{CRLB}}$ satisfies a scaling law upper bound:
	\vspace{0mm}
	\begin{equation}
		\lim\limits_{N \to \infty}\frac{E^*(N \times  n_0,\mathcal{A})}{E^*(n_0,\mathcal{A})} \,
		\le \mathcal{O}\bigl(N^{-2 \beta/d}\bigr).
		\vspace{0mm}
	\end{equation}
\end{theorem}

\begin{proof}
	Starting from the optimal placement of $n_0$ BSs in $\mathcal{A}$, we subdivide each of the $n_0$ original regions into $N$ congruent hypercubes by dividing each axis into $N^{1/d}$ equal segments (approximating $N^{1/d}$ as integer for large $N$, i.e., we can always find an integer \(Z\ge0\) with
	$N \;=\; Z^d$). This yields $Nn_0$ subregions, each with a volume of $|\mathcal{A}|/(N)$, and linear dimensions along each axis scaled by a factor $\kappa=N^{-1/d}$. 
	
	Let $E^*(n,\mathcal{A})$ denote the minimum achievable area-CRLB for the sensing-only deployment problem \eqref{P1.1} when $n$ BSs are deployed over region $\mathcal{A}$. 
	Based on Remark~\ref{RegularRegion}, we partition $\mathcal{A}$ into $N$ congruent subregions $\{\mathcal{A}_i\}_{i=1}^{N}$, and in each $\mathcal{A}_i$ we place $n_0$ BSs by taking the  optimal deployment for $\mathcal{A}$ and uniformly scaling it to fit $\mathcal{A}_i$. This construction yields a feasible deployment over the entire region $\mathcal{A}$ with a corresponding achieved objective value, denoted by ${E}_{\rm cons} (N \times n_0,\mathcal{A})$ for each subregion. By Proposition~\ref{ScaleInvariance}, scaling the region by a factor $\kappa$ results in the CRLB scaling by $\kappa^{2 \beta} $. Hence each new subregion's optimal CRLB is
	\vspace{0mm}
	\begin{equation}\label{ScalingProve}
		\begin{aligned}
			E^*(N \times n_0,\mathcal{A}) \le& {E}_{\rm cons} (N \times n_0,\mathcal{A}) \\
			\overset{({\rm{a}})}\le & E^*(n_0,\mathcal{A}_i) \\
			=& \kappa^{2 \beta} \,E^*(n_0,\mathcal{A}) 
			= N^{- 2\beta/d}\,E^*(n_0,\mathcal{A}).
		\end{aligned}	
	\end{equation}
	The inequality $({\rm a})$ holds due to the fact that in practice all $N \times n_0$ BSs can cooperate fully. Hence this partition-based result is indeed an upper bound, and the actual cooperative CRLB can be further reduced.
	According to (\ref{ScalingProve}), the mean CRLB over $(N \times n_0)$ BSs satisfies $E^*(Nn_0,\mathcal{A})
	\le N^{- 2\beta /d}\,E^*(n_0,\mathcal{A})$, i.e.,
	\vspace{0mm}
	\begin{equation}
		\lim\limits_{N \to \infty}\frac{E^*(Nn_0,\mathcal{A})}{E^*(n_0,\mathcal{A})}
		\le \mathcal{O}\bigl(N^{- 2\beta/d}\bigr).
		\vspace{0mm}
	\end{equation}
	This thus completes the proof.
\end{proof}

According to Theorem \ref{BSCountScaling}, when $\beta = 2$, the exponent $d =1,2,3$, yields the following explicit scaling law:
\vspace{0mm}
\begin{equation}
	\lim\limits_{N \to \infty}\frac{E^*(Nn_0,\mathcal{A})}{E^*(n_0,\mathcal{A})}
	\le \begin{cases}
		\mathcal{O}(N^{-4}), & d=1,\\
		\mathcal{O}(N^{-2}), & d=2,\\
		\mathcal{O}(N^{-4/3}), & d=3.
	\end{cases}
	\vspace{0mm}
\end{equation}
With BS count/density fixed, increasing the sensing-region dimension $d$ spreads nodes across more axes, shrinking the average solid-angle separation between BS look orientation and lengthening typical BS-target ranges; this dilutes bearing diversity, weakens the Fisher information, and reduces the marginal A-CRLB improvement per added BS.

\begin{remark}
	The above $\mathcal{O}(N^{-2\beta/d})$ bound in Theorem \ref{BSCountScaling} is relatively tight under dense BS environments: As $N$ grows, subregions more distant from a given target contribute vanishingly small error due to the increased distance aggravating the path‐loss. In particular, when $n_0$ is large so that each original BS region is already small, further subdividing introduces only marginal gains from distant cells, implying that the actual cooperative CRLB closely approaches this upper bound. Therefore, it follows that
	\vspace{0mm}
	\begin{equation}
		\lim\limits_{n_0 \to \infty} \lim\limits_{N \to \infty}\frac{E^*(N \times  n_0,\mathcal{A})}{E^*(n_0,\mathcal{A})} \,
		\approx \mathcal{O}\bigl(N^{-2\beta/d}\bigr)
		\vspace{0mm}
	\end{equation}
\end{remark}

\subsection{Low-complexity Sensing BS Deployment Optimization}
\label{LowComplixity}

Based on the Propositions and Theorem above, we can determine the optimal deployment plan for a new region by constructing a database of optimal configurations. By comparing the shape of the new region to those stored in the database, and using displacement, rotation, and scaling adjustments as necessary, one can readily construct deployment solutions.
Specifically, for any given deployment region \(\mathcal{A}\), there exists a set of representative regions \(\{\mathcal{A}_i\}_{i \in \cal{{I}}}\) for which the optimal deployments have already been solved, along with the corresponding solution set \(\{{\bf{b}}_n^*\}_{i \in \cal{{I}}}\). Then, there exists a transformation consisting of displacement, rotation, and scaling that maps \(\mathcal{A}_i\) to \(\mathcal{A}\), where $i$ denotes the index of the candidate area set $\cal{{I}}$. 
Then, the corresponding optimal deployment \({\bf{b}}^*\) can be derived using the inverse of the transformations applied above:  
\vspace{0mm}
\begin{equation}
	{\bf{b}}^* = \kappa \mathbf{O}_H \mathbf{R}_\varphi {\bf{b}}_{i}^* + \mathbf{\Delta}.
	\vspace{0mm}
\end{equation}  
Consequently, the optimal deployment for \(\mathcal{A}\) can be derived from \({\bf{b}}_n^*\) using the invariance and variance properties stated in Propositions \ref{TranslationInvariance}, \ref{RotationInvariance}, \ref{MirrorSymmetricRegion}, and \ref{ScaleInvariance}.  
This ensures that our invariance principles support an efficient and universal deployment strategy.

\section{ISAC Deployment Optimization}

In this section, we employ a MM framework, where each iteration constructs a surrogate upper bound and minimizes it, resulting in a robust, noise-agnostic procedure that can accommodate any design criterion.

\subsection{Problem Reformulation}
Directly integrating localization performance over a continuous region requires averaging over an uncountable set of target positions and random channel realizations, and this thus yields high-dimensional integrals that admit no closed-form solution. To handle the continuous‐region integration, we acquire \( K \) sampled points \( \{ \boldsymbol{t}_k \}_{k=1}^K \) uniformly distributed over the area of interest. Similarly, as for the communication constraint, the communication performance is evaluated at the sampled points drawn, ensuring that the area communication capacity meets the required threshold \( R^{th} \). 

Because the instantaneous channel gains \(g_n\) in (\ref{SNRexpression}) fluctuate randomly due to small-scale fading, it is impractical to adjust the BS placement relying on these transient values. 
To address the infeasibility of optimizing over random channel realizations, we first attain the average beamforming gain for both communication and sensing by taking the expectation over the small-scale fading, thereby ensuring a robust, long-term gain in both communication and sensing.  Then, we can derive the distribution of $g_{j}$ based on the moment matching technique of \cite{Hosseini2016Stochastic}. 
In (\ref{SNRexpression}), $g_{n}=\left|\mathbf{h}_{n}^H \mathbf{w}_n^c\right|^2$ represents the effective desired signals' channel gain, where $g_{n} \sim \Gamma \left( M_{\mathrm{t}} - 1, p^c\right)$ \cite{Hosseini2016Stochastic, Meng2024CooperativeTWC}. 
Similarly, the expected transmit beamforming gain for sensing can be formulated as ${\rm{E}}[\left|\mathbf{a}^H_{M_{\rm{t}}}(\Omega_n) \mathbf{w}_n^s\right|^2] = p^s\left( M_{\rm{t}} - 1\right)$, i.e., ${\rm{E}}[G_t] = p^s\left( M_{\rm{t}} - 1\right)$. 

Therefore, the ISAC deployment optimization problem can be formulated as
\vspace{0mm}
\begin{alignat}{2}
	\label{P2}
	(\rm{P2}): & \begin{array}{*{20}{c}}
		\mathop {\min }\limits_{\{{\bf{b}}_n\}} \quad  \frac{1}{K} \sum\nolimits_{k=1}^K {{\rm{CRLB}}\left( {\bf{t}}_k \right) } 
	\end{array} & \\ 
	\mbox{s.t.}\quad
	& \frac{1}{J} \sum\nolimits_{j=1}^J {R_c\left( {\bf{u}}_j \right) } \ge R^{th}. & \tag{\ref{P2}a}
	\vspace{0mm}
\end{alignat}
This sampling-based reformulation transforms the original continuous optimization problem into a tractable discrete form. The invariances in the preceding propositions and theorem carry over to the discretized formulation: co-transforming the sampling set with the region preserves sample-BS geometry, leaving the discrete objective and constraints unchanged. Hence the results of Section~\ref{DeploymentAnalysis} apply verbatim.

\subsection{Proposed Problem Solution}

Jointly optimizing multiple BS locations is challenging because the CRLB objective is highly non-convex. Plain gradient descent algorithm, without a majorizing surrogate, is markedly step-size-sensitive and often fails to achieve monotone descent. To tackle this challenge, we adopt a block-coordinate (alternating) MM scheme: at iteration \(t\), with \(N - 1\) BSs fixed, we majorize \(f(x)\) by a tight surrogate function \( g(x \mid x_t) \) and minimize it using a short projected gradient-descent step. This yields monotonic descent and convergence to a stationary point under standard MM conditions.

\subsubsection{MM Algorithm Framework}
In the majorization step, a surrogate function \( g(x \mid x_t) \) is constructed for globally upperbounding the objective function \( f(x) \) at point \( x_t \), satisfying $g(x \mid x_t) \geq f(x) \quad \text{and} \quad g(x_t \mid x_t) = f(x_t)$.
In the minimization step, the surrogate function is minimized to give the next update: $x_{t+1} = \arg \min_x g(x \mid x_t)$. 
This iterative process generates a non-increasing sequence \( \{ f(x_t) \}_{t \in \mathbb{N}} \), ensuring: $f(x_{t+1}) \leq g(x_{t+1} \mid x_t) \leq g(x_t \mid x_t) = f(x_t)$. 
Next, by optimizing each BS's location independently, while holding the remaining BS positions fixed, we obtain an initial coordinate set $\{x_i^{(0)}\}_{i=1}^N$.

\subsubsection{Problem Transformation}
\label{AlgorithmCVX}

Consider the optimization problem, where we aim for improving the localization performance, starting with an initialization that ensures having no node lies in the same plane as the target point. This initialization condition is crucial for achieving enhanced localization accuracy. In the following, we present a deployment optimization problem transformation method based on the MM framework. At each iteration, we optimize the position of one BS, while holding the other $(N-1)$ BSs fixed, thereby decomposing the original high-dimensional problem into a sequence of simpler subproblems. 

In the following, we introduce an optimization algorithm designed for updating the location of BS $n$ as a representative example. For notational simplicity, we first reformulate the problem for a single target point. The extension to the general case having 
$K$ target points can be achieved by summing the resultant objective function over all targets. To separate the term related to BS $n$ in the objective function, we define ${{\mathbf{v}}_{n,k}} = \frac{{{{{{\bf{t}}_k} - {\mathbf{b}_n}}}}}{{\left\| {{{{\bf{t}}_k}} - {\mathbf{b}_n}} \right\|}}$. Then the CRLB of target point ${\bf{t}}_k$, under centralized fusion-based localization, can be formulated as
 \vspace{0mm}
\begin{equation}\label{SensingLinkTermN}
\begin{aligned}
	&{\mathrm{CRLB}} ({\bf{t}}_k) = \text{Tr}\left(\mathbf{F}\left({\bf{t}}_k\right)^{-1}\right) \\
	&{\mathrm{Tr}}\bigg(\! \big(\underbrace{\underbrace{\frac{4{\mathbf{v}}_{n,k} {\mathbf{v}}_{n,k}^T}{\| {{\bf{b}}}_n - {\mathbf{t}}_k \|^{2 \beta}} }_{\text{monostatic link}}
	+ \underbrace{2 \sum_{\substack{i = 1, \\ i \neq n}}^{N} \! \frac{({\mathbf{v}}_{i,k} + {\mathbf{v}}_{n,k}) ({\mathbf{v}}_{i,k}^T + {\mathbf{v}}_{n,k}^T)}
	{\| {{\bf{b}}}_i - {\mathbf{t}}_k \|^{\beta} \| {{\bf{b}}}_n - {\mathbf{t}}_k \|^{\beta}} }_{\text{bi-static links}}}_{\text{Sensing links related to BS $n$}}
	+ \mathbf{M} \!\big)^{-1} \! \bigg),
\end{aligned}
\end{equation}
where ${\bf{M}} = \sum_{\substack{j = 1, j \neq n}}^{N}\sum_{\substack{i = 1,  i \neq n}}^{N} \frac{({\mathbf{v}}_{i,k} + {\mathbf{v}}_{j,k}) ({\mathbf{v}}_{i,k}^T + {\mathbf{v}}_{j,k}^T)}
{\| {{\bf{b}}}_i - {\mathbf{t}}_k \|^{\beta} \| {{\bf{b}}}_j - {\mathbf{t}}_k \|^{\beta}}$, including the term unrelated to the location of BS $n$. Assume that the initialization yields a non-degenerate BS-target geometry, so that the resulting matrix $\mathbf{M}$ is positive definite and hence invertible.

To facilitate problem transformation, let ${\bf{D}}_{n,k} = \mathrm{diag}\! \left(\! \frac{2}{{\left\| {\bf{b}}_1 - {\bf{t}}_k \right\|^{\beta} \left\| {\bf{b}}_n - {\bf{t}}_k \right\|^{\beta}}},\dots,\frac{1}{{\left\| {\bf{b}}_n - {\bf{t}}_k \right\|^{2 \beta}}},\dots,\frac{2}{{\left\| {\bf{b}}_N - {\bf{t}}_k \right\|^{\beta} \! \left\| {\bf{b}}_n - {\bf{t}}_k \right\|^{\beta}}} \! \!\right)$ $\in\mathbb{R}^{N\times N}$ and $
{\bf{U}}_{n,k} = \left[ {\mathbf{v}}_{1,k} + {\mathbf{v}}_{n,k}, \dots ,  {\mathbf{v}}_{N,k} + {\mathbf{v}}_{n,k}\right] \in\mathbb{R}^{3\times N}$. Then we can reformulate the sensing link terms in (\ref{SensingLinkTermN}) related to BS $n$ into ${\bf{U}}_{n,k}{\bf{D}}_{n,k}{{\bf{U}}}_{n,k}^T = \frac{4{\mathbf{v}}_{n,k} {\mathbf{v}}_{n,k}^T}{\| {{\bf{b}}}_n - {\mathbf{t}}_k \|^{2 \beta}} 
+ 2 \sum_{\substack{i = 1, i \neq n}}^{N} \frac{({\mathbf{v}}_{i,k} + {\mathbf{v}}_{n,k}) ({\mathbf{v}}_{i,k}^T + {\mathbf{v}}_{n,k}^T)}
{\| {{\bf{b}}}_i - {\mathbf{t}}_k \|^{\beta} \| {{\bf{b}}}_n - {\mathbf{t}}_k \|^{\beta}} $.
According to the Woodbury identity \cite{zhang2006schur}, it follows that
\vspace{0mm}
\begin{equation}
	\begin{aligned}
		&{\mathrm{CRLB}}({\bf{t}}_k) = {\mathrm{Tr}}\left(\left( {\bf{U}}_{n,k} {\bf{D}}_{n,k} {\bf{U}}_{n,k}^T + {\bf{M}} \right)^{-1}\right)  \\
		=& {\mathrm{Tr}}  \bigg( \! {\bf{M}}^{-1} \! -  \! {\bf{M}}^{-1} {\bf{U}}_{n,k} \! \left(\!  {\bf{D}}_{n,k}^{-1} + {\bf{U}}_{n,k}^T {\bf{M}}^{-1} \!\! {\bf{U}}_{n,k} \right)^{-1} \! {\bf{U}}_{n,k}^T {\bf{M}}^{-1}\bigg) \\
		=& {\mathrm{Tr}}\left({\bf{M}}^{-1} - {\bf{M}}^{-1} {\bf{U}}_{n,k} {\bf{Z}}_{n,k}^{-1} {\bf{U}}_{n,k}^T {\bf{M}}^{-1}\!\right),
		\vspace{0mm}
	\end{aligned}
\end{equation}
where $ {\bf{Z}}_{n,k} = {\bf{D}}_{n,k}^{-1} + {\bf{U}}_{n,k}^T {\bf{M}}^{-1} {\bf{U}}_{n,k}$. Then, the problem can be equivalently reformulated as maximizing ${\mathrm{Tr}}\left({\bf{M}}^{-1} {\bf{U}}_{n,k} {\bf{Z}}_{n,k}^{-1} {\bf{U}}_{n,k}^T {\bf{M}}^{-1}\right)$. To further facilitate the problem solution, the following inequality is introduced:
\vspace{0mm}
\begin{equation}\label{Inequality1}
	\begin{aligned}
	&- \operatorname{Tr} \left( {\mathbf{M}}^{-1} \mathbf{U}_{n,k} \mathbf{Z}_{n,k}^{-1} \mathbf{U}_{n,k}^T {\mathbf{M}}^{-1} \right) \\
	\leq & - \operatorname{Tr} \left[ 2 \left( {\mathbf{M}}^{-2} \mathbf{U}_{n,k}^{(r)} (\mathbf{Z}_{n,k}^{(r)})^{-1} \right)^T \left( \mathbf{U}_{n,k} - \mathbf{U}_{n,k}^{(r)} \right) \right] \\
	& +  \!\operatorname{Tr}  \!\left[ \! \left( \! (\mathbf{Z}_{n,k}^{(r)})^{-1} (\mathbf{U}_{n,k}^{(r)})^T {\mathbf{M}}^{-2} \mathbf{U}_{n,k}^{(r)} (\mathbf{Z}_{n,k}^{(r)})^{-1} \right)^T \! \! ( \mathbf{Z}_{n,k} - \! \mathbf{Z}_{n,k}^{(r)} )  \! \right] \\
	& - \operatorname{Tr} \left( {\mathbf{M}}^{-1} \mathbf{U}_{n,k}^{(r)} (\mathbf{Z}_{n,k}^{(r)})^{-1} (\mathbf{U}_{n,k}^{(r)})^T {\mathbf{M}}^{-1} \right),
\end{aligned}
\end{equation}
where $\mathbf{U}_{n,k}^{(r)}$ and $\mathbf{Z}_{n,k}^{(r)}$ represent the variable value at the $r$th iteration.
According to (\ref{Inequality1}), removing the constant part, the objective function can be approximated as
\vspace{0mm}
\begin{equation}
	\label{Simplified1equation}
	\mathop {\min }\limits_{{\bf{b}}_n} \quad \operatorname{Tr} \left( \mathbf{G}_{n,k} \mathbf{Z}_{n,k} - \mathbf{P}_{n,k}^T \mathbf{U}_{n,k} \right),
\end{equation}
where $\mathbf{G}_{n,k} =  \left((\mathbf{Z}_{n,k}^{(r)})^{-1} (\mathbf{U}_{n,k}^{(r)})^T {\mathbf{M}}^{-2} \mathbf{U}_{n,k}^{(r)} (\mathbf{Z}_{n,k}^{(r)})^{-1}\right)^{T} $ and $\mathbf{P}_{n,k} = 2 {\mathbf{M}}^{-2} \mathbf{U}_{n,k}^{(r)} (\mathbf{Z}_{n,k}^{(r)})^{-1}$.
Then, by substituting ${\bf{Z}}_{n,k}$ into (\ref{Simplified1equation}), the objective function in (\ref{Simplified1equation}) can be expressed as $
\operatorname{Tr} \left( \mathbf{G}_{n,k} {\mathbf{D}}_{n,k}^{-1} \right) + \operatorname{Tr} \left( \mathbf{G}_{n,k}  \mathbf{U}_{n,k}^T {\mathbf{M}}^{-1} \mathbf{U}_{n,k} \right) - \operatorname{Tr} \left( \mathbf{P}_{n,k}^T \mathbf{U}_{n,k} \right)$, where $\operatorname{Tr} \left( \mathbf{G}_{n,k} {\mathbf{D}}_{n,k}^{-1} \right)$ comprises quadratic and quartic functions of the BS location vectors. By introducing an appropriate change of variables $\delta_{n,k}$ satisfying $\left\| \mathbf{t}_k - {\bf{b}}_n \right\|^{\beta} \le \delta_{n,k}, \forall n,k$, it can be rendered convex. Then, we can first handle the term $\operatorname{Tr} \left( \mathbf{G}_{n,k}  \mathbf{U}_{n,k}^T {\mathbf{M}}^{-1} \mathbf{U}_{n,k} \right)$ as follows.
Define $\mathbf{U}_{n,k}^c = \left[ \mathbf{v}_{1,k}, \dots, \mathbf{v}_{n-1,k}, 0, \mathbf{v}_{n+1,k}, \dots, \mathbf{v}_{N,k}\right]
$, we have $\mathbf U_{n,k}= \mathbf U_{n,k}^c + \mathbf v_{n,k}\mathbf 1^T + \mathbf v_{n,k}\mathbf e_n^T$. Then, by substituting ${\bf{Z}}_{n,k}$ into the above equation and separating the linear and quadratic contributions in $\mathbf{v}_{n,k}$, the objective function can be transformed into
\vspace{0mm}
\begin{equation}
\begin{aligned}
	&\text{Tr} \left( {\mathbf{G}}_{n,k} \mathbf{U}_{n,k}^T {\mathbf{M}}^{-1} \mathbf{U}_{n,k} \right) \\
	=& a_1 \mathbf{v}^T_{n,k} {\mathbf{M}}^{-1} \mathbf{v}_{n,k} \! + \! \text{Tr} \left( {\bf{q}}_1^T \mathbf{v}_{n,k} \right) \! + \! \text{Tr} \left(  \mathbf{G}_{n,k} (\mathbf{U}_{n,k}^c)^T  {\mathbf{M}}^{-1} \mathbf{U}_{n,k}^c \right)  \!,
	\vspace{0mm}
\end{aligned}
\end{equation}
where $a_1 = \mathbf{1}^T \mathbf{G}_{n,k} \mathbf{1}  + \mathbf{e}_n^T \mathbf{G}_{n,k} \mathbf{1} + \mathbf{1}^T \mathbf{G}_{n,k} \mathbf{e}_n + \mathbf{e}^T_n \mathbf{G}_{n,k} \mathbf{e}_n$,  and ${\bf{q}}_1^T = \left( {\mathbf{M}}^{-1} \mathbf{U}_{n,k}^c \mathbf{G}_{n,k} \mathbf{1} \right)^T 	+ \mathbf{1}^T \mathbf{G}_{n,k} (\mathbf{U}^c_{n,k})^T {\mathbf{M}}^{-1}+ \mathbf{e}^T_n \mathbf{G}_{n,k} (\mathbf{U}_{n,k}^c)^T {\mathbf{M}}^{-1} + \left( {\mathbf{M}}^{-1} \mathbf{U}_{n,k}^c \mathbf{G}_{n,k} {\mathbf{e}}_n \right)^T$.

However, $a_1 \mathbf{v}_{n,k}^T {\bf{M}}^{-1} \mathbf{v}_{n,k}$ is non-convex, because $\mathbf{v}_{n,k}$ introduces a nonlinearity involving the reciprocal of the Euclidean distance between the BS and the target. Thereby, we have to further transform it based on the following inequality:
\vspace{0mm}
\begin{equation}
	\begin{aligned}
		&{\mathrm{Tr}}\left( a_1 \mathbf{v}_{n,k}^T {\mathbf{M}}^{-1} \mathbf{v}_{n,k} \right) 	\le 2a_1 \left( \tilde{\mathbf{ M}} \mathbf{v}_{n,k}^{(r)} \right)^T \left( \mathbf{v}_{n,k} - \mathbf{v}_{n,k}^{(r)} \right) \\
		& + a_1\lambda_{\max} \left( {\mathbf{M}}^{-1} \right)  +   a_1\left( \mathbf{v}_{n,k}^{(r)} \right)^T {\mathbf{\tilde M}} \mathbf{v}_{n,k}^{(r)} ,
		\vspace{0mm}
	\end{aligned}
\end{equation}
where $\tilde {\mathbf{ M}} = {\mathbf{M}}^{-1} - \lambda_{\max} \left( {\mathbf{M}}^{-1} \right) {\mathbf{I}}$, and $\lambda_{\max}({\mathbf{M}}^{-1})$ denotes the maximum eigenvalue of  matrix ${\mathbf{M}}^{-1}$. In addition, in (\ref{Simplified1equation}), the term $\mathbf{P}_{n,k}^T \mathbf{U}_{n,k} $ can be reduced by discarding any components independent of the optimization variable ${\bf{b}}_n$, leaving only the contributions $({{{\mathbf{1}}^T}{{\mathbf{P}}_{n,k}^T} + {{\mathbf{e}}^T_n}{{\mathbf{P}}_{n,k}^T}}) \mathbf{v}_{n,k}$. 
By omitting the constant terms, the objective function can be approximated as 
\vspace{0mm}
\begin{equation}\label{NonconvexApprox}
	\mathop {\mathrm{min}}\limits_{{\mathbf{b}}_n} \frac{1}{K} \sum\nolimits_{k = 1}^K \left( a_2 \delta^2_{n,k} + a_3 \delta_{n,k} +  {\mathbf{q}}_2^T \frac{\mathbf{t}_k - {\mathbf{b}}_n}{\left\| \mathbf{t}_k - {\mathbf{b}}_n \right\|} \right) ,
	\vspace{0mm}
\end{equation}
where ${\mathbf{q}}_2^T = 2a_1 \left( \tilde{\mathbf{ M}} \mathbf{v}_{n,k}^{(r)} \right)^T + {\mathbf{q}}_1^T + {{{\mathbf{1}}^T}{{\mathbf{P}}_{n,k}^T} + {{\mathbf{e}}^T_n}{{\mathbf{P}}_{n,k}^T}}$, and  $a_2 = {\mathbf{G}}_{n,k}{[n,n]}$, $a_3 = \frac{1}{2} \sum\nolimits_{m = 1, m \ne n}^N {\mathbf{G}}_{n,k}{[m,m]}$  $\left\| \mathbf{t}_k - {\mathbf{b}}_{m} \right\|^{\beta}$.
Here, (\ref{NonconvexApprox}) is still non-convex due to  $\frac{\mathbf{t}_k - {\mathbf{b}}_n}{\left\| \mathbf{t}_k - {\mathbf{b}}_n \right\|}$. By utilizing the trust region algorithm of \cite{yuan2015recent}, we introduce a constraint $\| {\mathbf{b}}_n - {\mathbf{b}}_n^{(r)} \| \le \epsilon$ to limit the update region. In the trust region method, a first-order approximation is used at each iteration for constructing a local linear model of the objective function based on its gradient. In practice, this leads to computing the so-called Cauchy point \cite{yuan2015recent}, which is obtained by taking a step in the negative gradient's direction scaled to the boundary of the trust region. 
Specifically, we adopt Taylor expansion to approximate the objective function as follow.
\vspace{0mm}
\begin{equation}
		\frac{\mathbf{t}_k - {\mathbf{b}}_n}{\left\| \mathbf{t}_k - {\mathbf{b}}_n \right\|} 
		 \approx  \frac{\mathbf{t}_k - {\mathbf{b}}_n^{(r)}}{\|\mathbf{t}_k - {\mathbf{b}}_n^{(r)}\|} + g'({\mathbf{b}}^{(r)}_n) ({\mathbf{b}}_n - {\mathbf{b}}_n^{(r)}),
		 \vspace{0mm}
\end{equation}
where $g'(\mathbf b^{(r)}_n)
= -\frac{1}{\|\mathbf t_k - \mathbf b_n^{(r)}\|}\,
\Biggl(\,
\mathbf I
- \frac{(\mathbf t_k - \mathbf b_n^{(r)})(\mathbf t_k - \mathbf b_n^{(r)})^T}
{\|\mathbf t_k - \mathbf b_n^{(r)}\|^2}
\Biggr)$. Following the transformation above, every term in the objective function involving the decision variables is either a convex quadratic or a linear function.
Then, for the communication constraint, we introduce auxiliary variables:
\vspace{0mm}
\begin{equation}\label{ReplacementConstraint}
	z_{n,j} \ge \|{\mathbf{b}}_n - {\mathbf u}_j\|^{\alpha}, \quad \forall j.
	\vspace{0mm}
\end{equation}
We define $C_j = 1 + \frac{1}{\tilde \sigma_c^2} \sum_{m \ne n} g_m \| \mathbf{b}_m - \mathbf{u}_j \|^{-\alpha}$,
and apply the first-order Taylor expansion to approximate the communication rate of the user located at ${\bf{u}}_j$ as $R_j \ge \log_2\left(C_j + \frac{g_n}{\tilde \sigma_c^2 z_{n,j}^{(r)}} \right) 
- \frac{ {g_n}{} }{\tilde \sigma_c^2 (z_{n,j}^{(r)})^2 \left( C_j + \frac{g_n}{\tilde \sigma_c^2 z_{n,j}^{(r)}} \right) \ln 2 } 
\left( z_{n,j} - z_{n,j}^{(r)} \right) = \tilde R_j$.
Thus, the original non-convex communication constraint can be approximately replaced by the following convex constraint:
\vspace{0mm}
\begin{equation}\label{ConvexCommConstraint}
	\frac{1}{J} \sum\nolimits_{j=1}^J \tilde R_j
	\ge R^{\mathrm{th}}.
	\vspace{0mm}
\end{equation}
Together with the convex constraint \eqref{ReplacementConstraint}, this yields a convex approximation with respect to the variables \( \{{\mathbf{b}}_n, z_{n,j} \} \). By iteratively updating the expansion point \( z_{n,j}^{(r)} = \|{\mathbf{b}}_n^{(r)} - {\mathbf u}_j \|^\alpha \), the algorithm converges to a local optimum of the original problem.

Finally, the problem can be approximated by the following convex program:
\vspace{0mm}
\begin{alignat}{2}
	\label{P3}
	(\mathrm{P3}): \quad & \min_{{\mathbf{b}}_n, \delta_{n,k}, z_{n,j}} 
	\frac{1}{K} \sum_{k = 1}^K \left( a_2 \delta_{n,k}^2 + a_3 \delta_{n,k} + \mathbf q_{2}^T g'({\mathbf{b}}_n^{(r)})  {\mathbf{b}}_n \right) \\ 
	\text{s.t.} \quad
	& \eqref{ReplacementConstraint}, \eqref{ConvexCommConstraint}, & \tag{\ref{P3}a} \\
	& \| {\mathbf{b}}_n - {\mathbf{b}}_n^{(r)} \| \le \epsilon, & \tag{\ref{P3}b} \\
	& \| \mathbf{t}_k - {\mathbf{b}}_n \|^{\beta} \le \delta_{n,k}, \forall k. & \tag{\ref{P3}c}
\end{alignat}
As a result, the approximated  problem can be updated and solved in an iterative manner by existing convex optimization solvers.
After solving the subproblem constructed for moving a step \({\mathbf{\Delta b}}_r\), we quantify how well the approximate objective function \(f_r({\mathbf{b}}_n) = \frac{1}{K} \sum\nolimits_{k = 1}^K \left( a_2 \delta^2_{n,k} + a_3 \delta_{n,k} + \mathbf q_{2}^T g'({\mathbf{b}}^{(r)}_n)  {{\mathbf{b}}_n} \right) \) predicted the actual reduction in \(f\):  
\vspace{0mm}
\begin{equation}\label{ComputateRatio}
	\rho_r
	= \frac{{\mathrm{CRLB}}\bigl({\mathbf{b}}_n^{(r)}\bigr)\;-\;{\mathrm{CRLB}}\bigl({\mathbf{b}}_n^{(r+1)}\bigr)}
	{{\mathrm{CRLB}}\bigl({\mathbf{b}}_n^{(r)}\bigr)\;-\;f_r({\mathbf{b}}_n^{(r+1)})},
	\vspace{0mm}
\end{equation}
where \(\mathrm{CRLB}\) is the true objective (from problem (P3)) evaluated at the current BS position, and \(f_r({\mathbf{b}}_n^{(r+1)})\) is its local approximation in (P3). Then, we have to accept or reject the step. If \(\rho_r\ge\eta_s\), the step is deemed sufficiently accurate and we set  
$
{\mathbf{b}}_n^{(r+1)} = {\mathbf{b}}_n^{(r)} + {\mathbf{\Delta b}}_r.
$ 
Otherwise, we reject the step and keep  
$
{\mathbf{b}}_n^{(r+1)} = {\mathbf{b}}_n^{(r)}.
$
In addition, we also have to adjust the trust‐region radius. 
Based on \(\rho_r\), we expand, keep, or shrink the region in which our quadratic model is trusted: 
\vspace{0mm} 
\begin{equation}\label{UpdateE}
\epsilon_{r+1} =
\begin{cases}
	\gamma_n\,\epsilon_r,      & \rho_r \ge \eta_v,\\
	\epsilon_r,                 & \eta_s \le \rho_r < \eta_v,\\
	\gamma_d\,\epsilon_r,      & \rho_r < \eta_s.
\end{cases}
\vspace{0mm}
\end{equation}
Here \(\eta_s<\eta_v\) are acceptance thresholds, where \(\gamma_n>1\) expands the region when the model is reliable, while \(\gamma_d\in(0,1)\) contracts it when the model poorly predicts \(f\).

\begin{algorithm}[t] 
	\footnotesize
	\caption{Proposed MM-based Algorithm}
	\label{Algorithm}
	\begin{algorithmic}[1]
		\STATE {\bf{Initialize}} BS positions $\{{\bf{b}}_n^{(0)}\}$, sensing target positions $\{{\bf{t}}_k\}$,  communication user positions $\{{\bf{u}}_j\}$, and trust region radius $\epsilon_0$.
		\REPEAT
		\FOR{each BS $n = 1,\dots,N$}
		\STATE Solve problem (P3) to obtain optimal movement $\Delta\mathbf{b}_n^{(r)}$ for BS $n$.
		\STATE Compute predicted and actual reduction ratio $\rho_r$ based on (\ref{ComputateRatio}).
		\IF {$\rho_r \ge \eta_s$}
		\STATE Accept update: ${\mathbf{b}}_n^{(r+1)} = {\mathbf{b}}_n^{(r)} + \Delta\mathbf{b}_n^{(r)}$
		\ELSE
		\STATE Reject update: ${\mathbf{b}}_n^{(r+1)} = {\mathbf{b}}_n^{(r)}$
		\ENDIF
		\STATE Update trust-region parameters according to (\ref{UpdateE})
		\ENDFOR
		\UNTIL{convergence of the optimal $\{{\mathbf{b}}_n^*\}$}
	\end{algorithmic}
\end{algorithm}

\section{System Performance}

Using numerical results in this section, we have studied the fundamental characteristics of ISAC networks and verified the tightness of the expression derived by Monte Carlo simulation results. The system parameters are given based on empirical and modeling studies \cite{Sun2016Investigation}. Specifically, the number of transmit antennas is $M_{\mathrm{t}} = 5$; the number of receive antennas is $M_{\mathrm{r}} = 5$; the transmit power is $P_{\mathrm{t}} = 0.01$W at each BS; $|\xi|^2 = 10$; the frequency $f^c = 5$ GHz; the bandwidth $B \in [10, 100]$ MHz; the noise power $-120$dB; pathloss coefficients $\alpha = 4$ and $\beta = 2$. In the following experiments, we consider both one-dimensional and two-dimensional target and user regions. It is important to note that in all scenarios, the BS locations are optimized within a three-dimensional deployment space. In the simulations, we assume that the target and users are located within the same region and are sampled at the same rate.

\subsection{Convergence and Scaling Law Evaluation}

\begin{figure}[t]
	\centering
	\includegraphics[width=8.4cm]{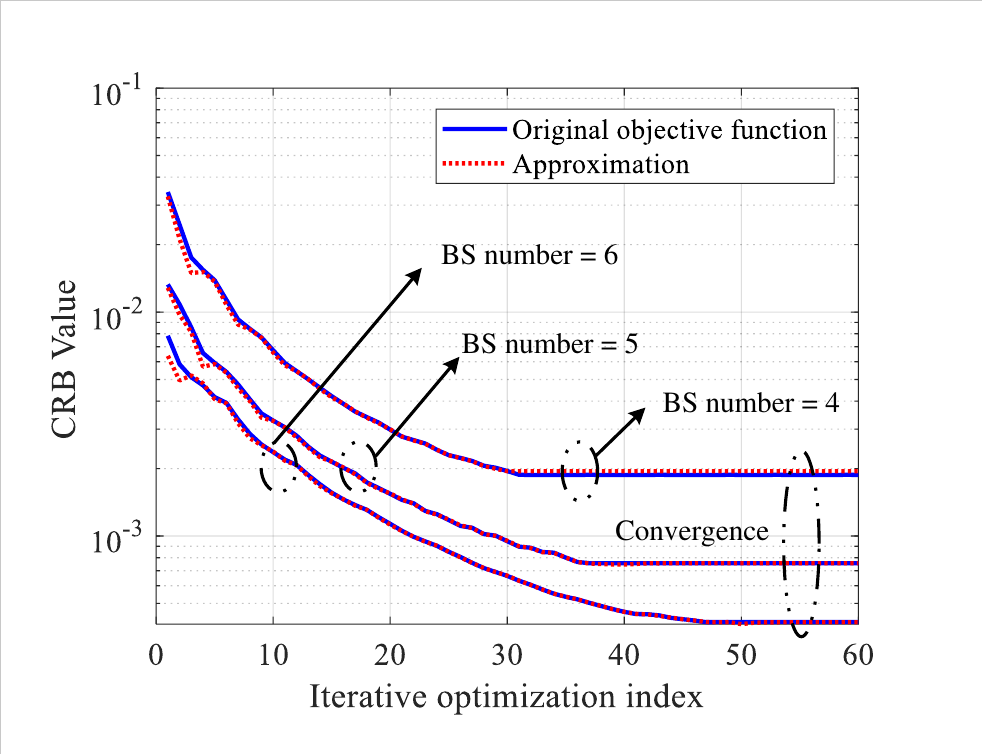}
	\vspace{0mm}
	\caption{Convergence evaluation of the proposed MM-based algorithm.}
	\label{figure5}
\end{figure}
Fig. \ref{figure5} illustrates the convergence behavior of our proposed algorithm. When the number of BSs is four, the algorithm converges in approximately 30 iterations. When the number of BSs increases, the number of iterations required grow correspondingly due to the enlarged solution space. Moreover, by comparing the values of our surrogate function against those of the original objective function, we observe that the surrogate function closely tracks the true objective function and they ultimately converge to the same value, thereby confirming the efficiency and accuracy of the proposed algorithm.

\begin{figure}[t]
	\centering
	\includegraphics[width=8.4cm]{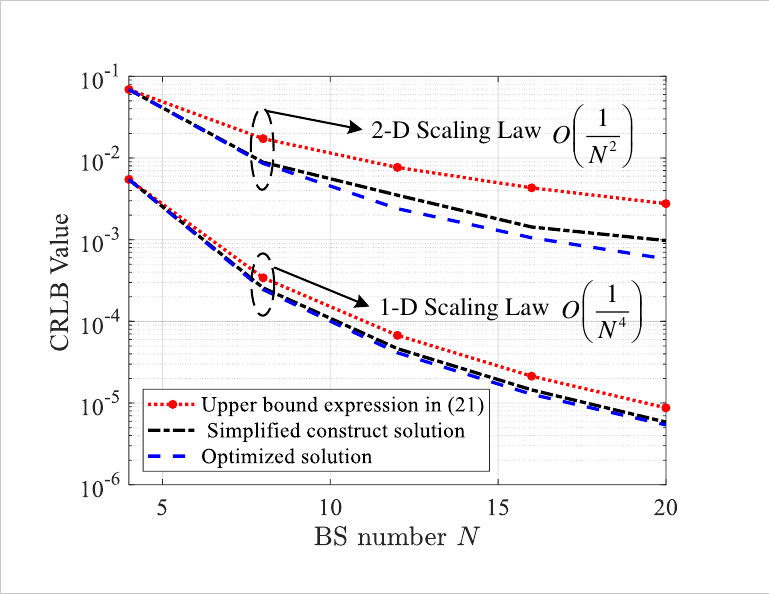}
	\vspace{0mm}
	\caption{Evaluation of scaling law derived in Theorem \ref{BSCountScaling} with 1D and 2D sensing area.}
	\label{figure10}
\end{figure}

\begin{figure}[htbp]
	\centering
	\vspace{0mm}
	\subfigure[BS deployment for optimal communication achieving $R^*$.]{
		\label{figure2a}
		\includegraphics[width=7.5cm]{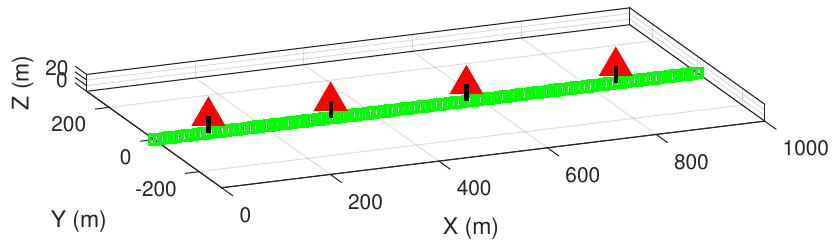}
	}
	\subfigure[BS deployment for optimal sensing.]{
		\label{figure2b}
		\includegraphics[width=7.5cm]{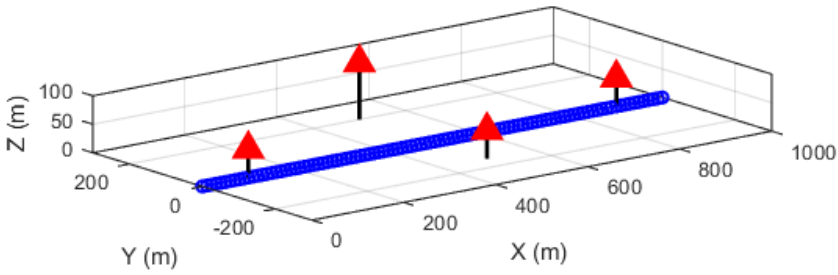}
	}
	\subfigure[BS deployment for maximizing sensing performance with $R_{th} = 0.8 R^*$.]{
	\label{figure2c}
	\includegraphics[width=7.5cm]{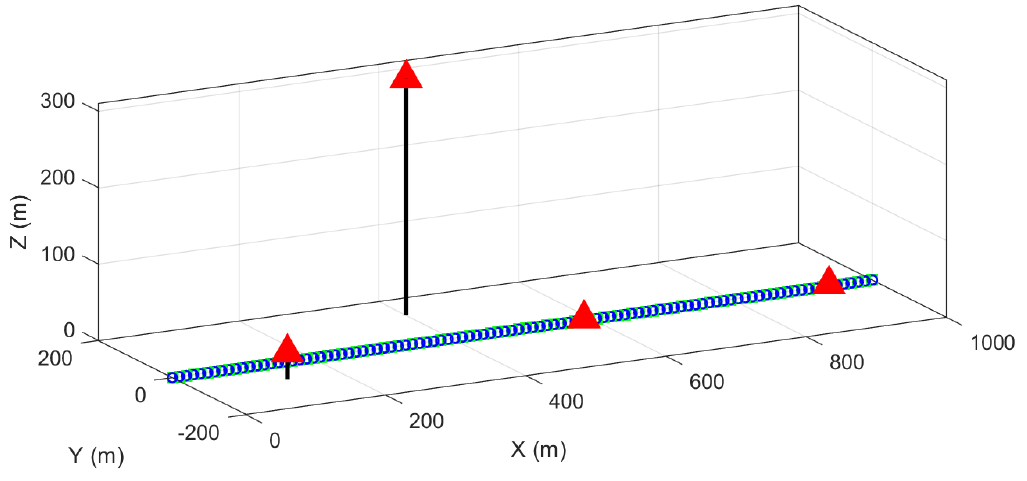}
	}
	\caption{BS deployment under varying sensing-communication weights.}
	\label{figure2ab}
\end{figure}

\begin{figure}[t]
	\centering
	\includegraphics[width=8.4cm]{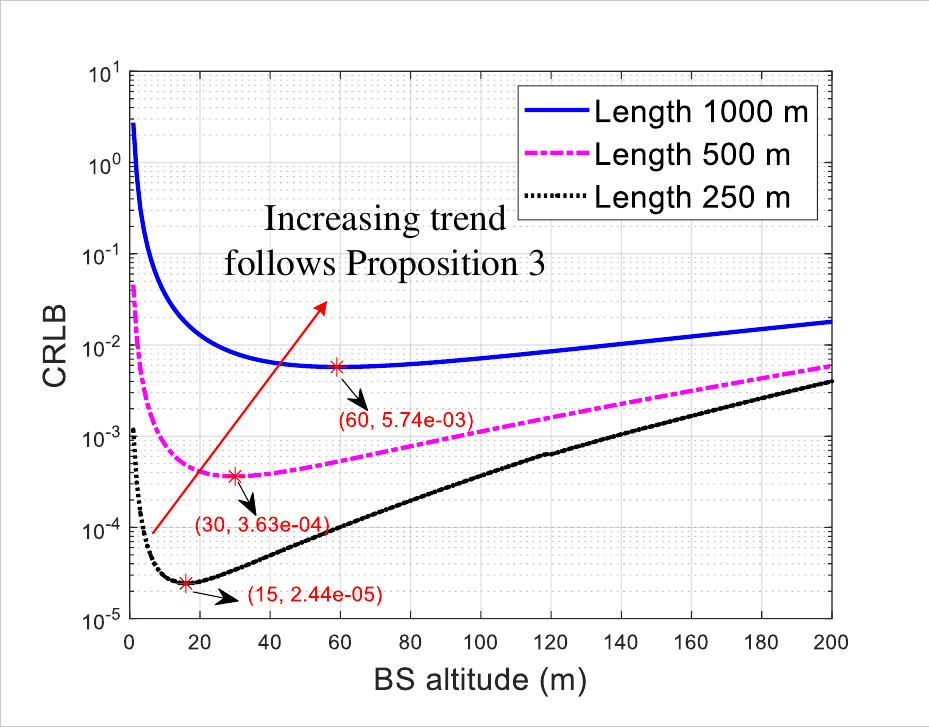}
	\vspace{0mm}
	\caption{Sensing performance versus BS height.}
	\label{figure6}
\end{figure}

Fig. \ref{figure10} presents the scaling law of the area CRLB as a function of the number of BSs, clearly showing a reduction in CRLB upon increasing the BS number. To ensure a relatively accurate scaling curve, we employ iterative optimization, where at each step we perform a global search to reposition one BS, while holding all others fixed, and repeat this process for each BS until convergence. Moreover, the scaling gain erodes as the sensing‐area dimension increases. More explicitly, under the same multiplicative increase in BS number, a 2D region must spread nodes over two orthogonal axes, whereas a 1D region aligns them along a single axis. To validate the periodic‐pattern deployment discussed in Remark \ref{RegularRegion}, we took the optimized four-BS layout and replicated it across the entire enlarged region, namely 'simplified constructed solution' in Fig.  \ref{figure10}. The resultant performance closely matches that of the fully optimized scheme, especially in the 1D scenario, demonstrating the practical effectiveness of regular deployments in large-scale networks. Furthermore, in 1D scenarios the theoretical CRLB‐scaling bound closely aligns with the simulation results, since the BSs deployed across the $n_0$ fundamental unit regions already realize the majority of the localization gain achievable under full BS cooperation. By contrast, in 2D scenarios a persistent gap arises between the theoretical bound and simulation results. The simplified upper bound does not capture the sensing‐performance degradation that occurs at the peripheries of individual subregions. By introducing additional BSs that jointly cover the entire 2D area, especially along edges and corners, these boundary effects are substantially alleviated, yielding localization gains beyond those predicted by the bound. Consequently, as the BS density increases, the empirical CRLB scaling in 2D surpasses the theoretical prediction.

\subsection{Sensing  Performance of 1D Area}

To contrast BS deployments optimized for area‐sensing performance to those tailored for communication coverage, we first consider an elementary one‐dimensional user region along the line \(y=0\), \(x\in[0,1000]\).  Users (green rectangle) are placed uniformly along this line, and a fixed number of BSs (red triangles) are placed according to two distinct strategies. Targets (blue dots) are likewise interspersed along the same line to evaluate sensing performance. See Figs. \ref{figure2ab}, \ref{figure4abs}, \ref{figure9ab}, and \ref{figure8abs} for illustration.
In Fig.~\ref{figure2a}), all BSs lie as close as possible to the communication user line, i.e., above \(y=0\) with their \(x\)-coordinates evenly spaced over \([0,1000]\). This arrangement minimizes the path-loss for each BS-user link, yielding maximal end‐to‐end SNR for purely communication-centric operation.
By contrast, in Fig.~\ref{figure2b}), to balance the associated sensing signal power and geometry gain requirements, two BSs remain near the target line (at \(y\approx0\)) to preserve low‐attenuation links, while the remaining BSs are shifted in \(y\) to form near‐perpendicular baselines relative to the target line.  This bi‐static geometry, yielding approximately \(90^\circ\) incidence angles at most targets enhances geometric gain, thereby improving area‐wide sensing accuracy without unduly sacrificing communication performance. As shown in Fig. \ref{figure2c}, under the communication constraint the optimal deployment places several BSs near the user region to boost link SNR, while positioning one BS at higher altitude to enhance sensing geometry over the target area, i.e., to increase angular diversity.

\begin{figure*}[htbp]
	\centering
	\vspace{0mm}
	\subfigure[Deployment for maximizing $R_c$.]{
		\label{figure3a}
		\includegraphics[width=5.7cm]{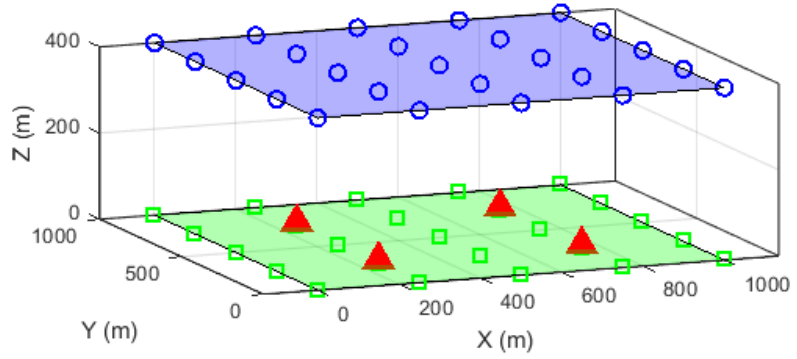}
	}
	\subfigure[Deployment for minimizing CRLB, $R_c \!\ge \! 5$ bps/Hz.]{
		\label{figure3b}
		\includegraphics[width=5.7cm]{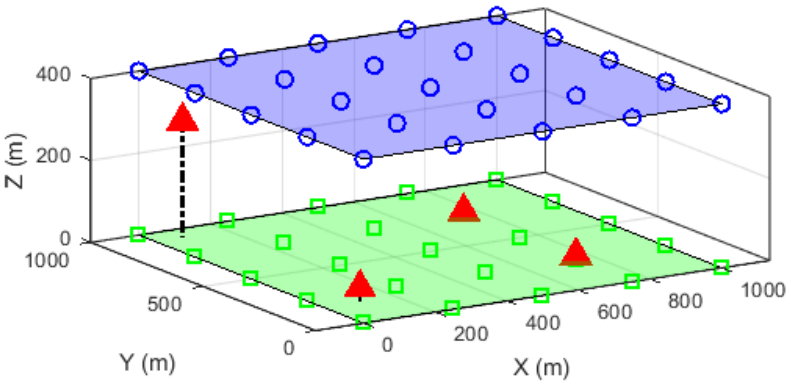}
	}
	\subfigure[Deployment for minimizing CRLB.]{
		\label{figure3c}
		\includegraphics[width=5.4cm]{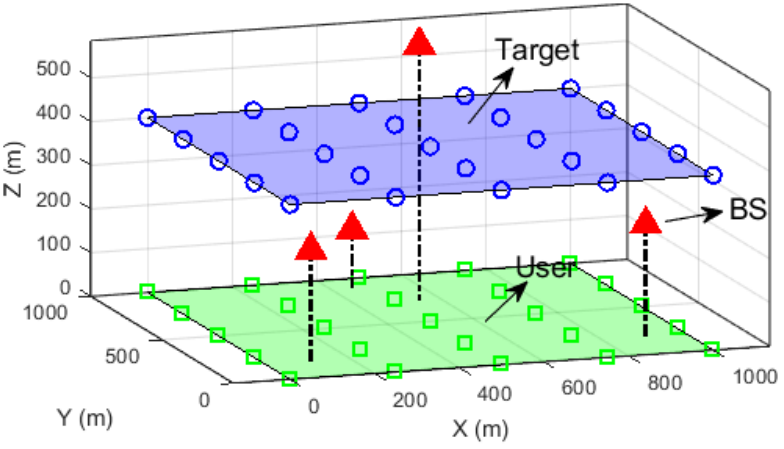}
	}
	\subfigure[CRLB map for communication optimal.]{
		\label{figure4a}
		\includegraphics[width=5.7cm]{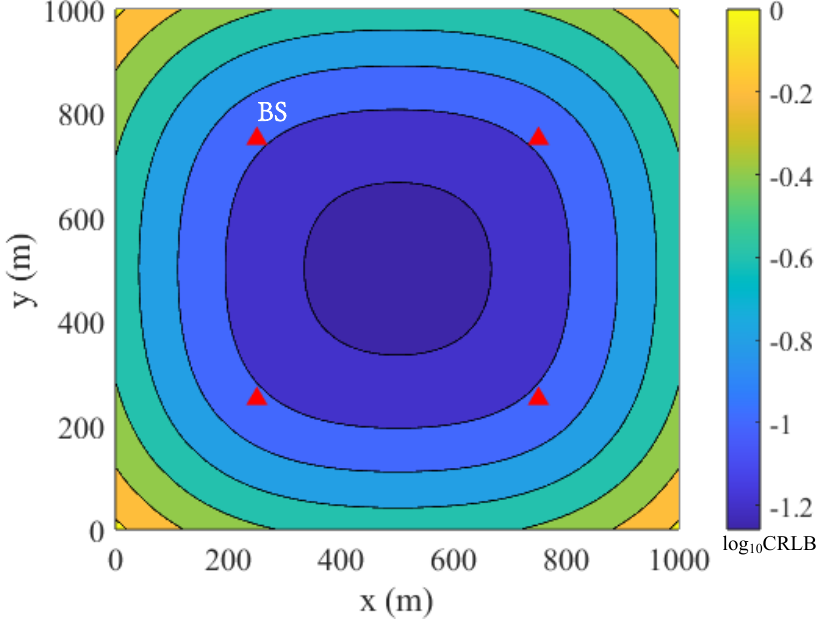}
	}
	\subfigure[CRLB map with $R_c \ge 5$ bps/Hz.]{
		\label{figure4b}
		\includegraphics[width=5.1cm]{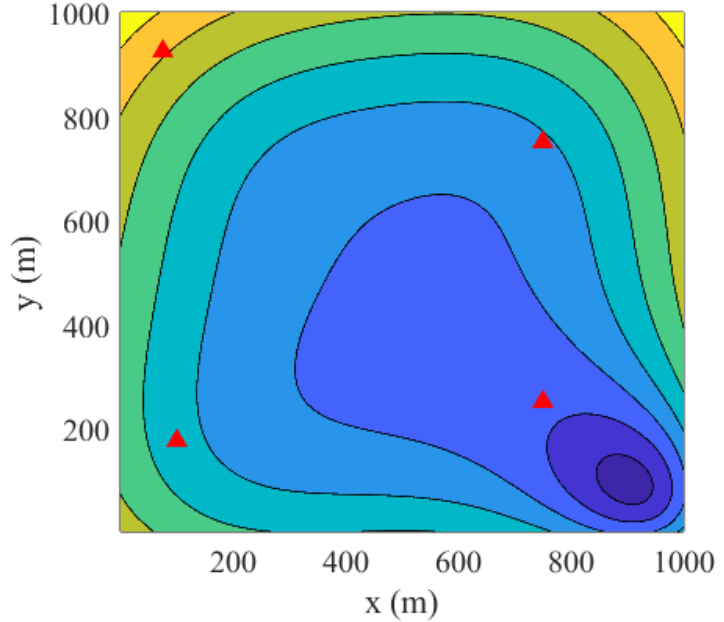}
	}
	\subfigure[CRLB map for sensing optimal.]{
		\label{figure4c}
		\includegraphics[width=5.1cm]{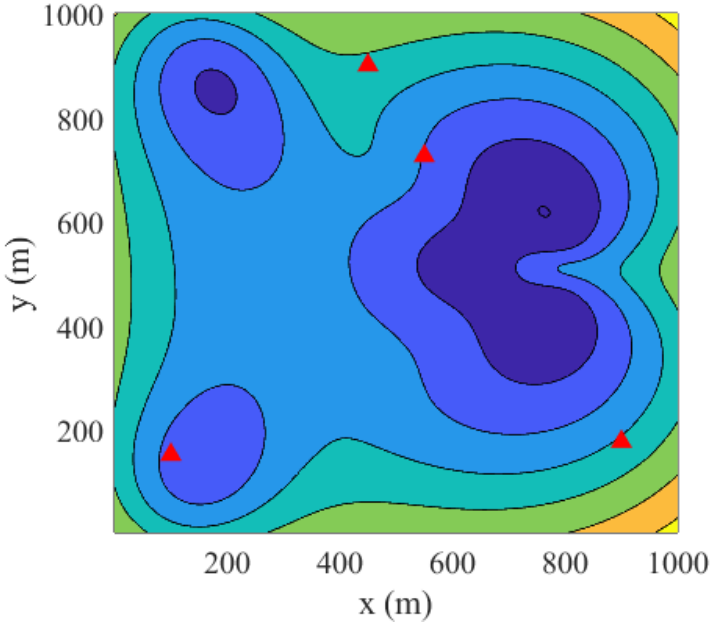}
	}
\vspace{0mm}
	\caption{Comparisons between sensing and communication (The color bar shows the base-10 logarithm of the CRLB, and same for the following results).}
	\label{figure4abs}
\end{figure*}

Fig.~\ref{figure6} shows the minimal CRLB as a function of the one‐dimensional region length $L$ and various BS heights. In each curve, the region length is increased from 250m to 1000m, and the optimal BS height is selected to minimize the averaged CRLB over all target positions in the interval \([0,L]\).
Moreover, when \(L\) is multiplied by an integer factor, the optimal BS height \(h^*\) also increases by roughly the same factor. This one‐to‐one proportionality between region size and BS height is a direct consequence of the geometric invariance properties. Explicitly, scaling the coverage interval by \(\kappa\) only preserves the angular geometry if the BS elevation is scaled by \(\kappa\) as well, as analyzed in Proposition \ref{ScaleInvariance}.
The dominant effect driving the optimal BS height is the geometric gain in target sensing. Increasing the number of BSs in proportion to the region length maintains consistent incidence angles and range spreads, thereby balancing the multiple impact factors in the CRLB across the entire coverage area.  These numerical findings validate our theoretical derivations and provide concrete guidelines for selecting the BS heights in large‐scale ISAC deployments.
\begin{figure}[t]
	\centering
	\includegraphics[width=8.4cm]{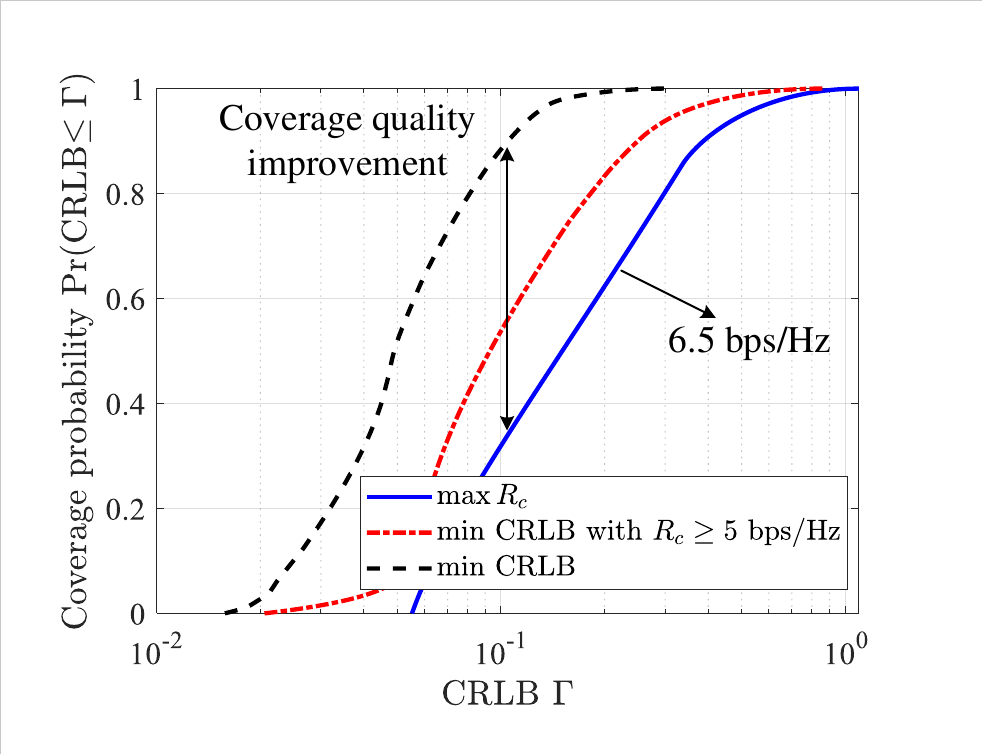}
	\vspace{0mm}
	\caption{CRLB coverage probability under various objectives and constraints.}
	\label{figure4}
\end{figure}

\subsection{Sensing  Performance of 2D Area}

As illustrated in Fig.~\ref{figure4abs}, we assess the two‐dimensional coverage performance of communication and sensing services, where the communication users are uniformly distributed on the ground plane (green) and targets over a parallel plane at an altitude of 400m (blue). Figs. \ref{figure3a} and \ref{figure4a}, \ref{figure3b} and \ref{figure4b}, as well as \ref{figure3c} and \ref{figure4c} depict three deployment strategies maximizing communication rate, minimizing the CRLB under a 5 bps/Hz communication constraint, and minimizing the CRLB without any communication requirement, respectively. In the communication‐only scenario (c.f. Fig.~\ref{figure3a} and Fig.~\ref{figure4a}), BSs are arranged symmetrically across the area, yielding a symmetric CRLB map with values ranging from $10^{-1.2}$ to $10^{0} = 1$. When enforcing the 5 bps/Hz rate (c.f. \ref{figure3b} and \ref{figure4b}), the BSs shift slightly toward the sensing plane and adjust their planar geometry to balance both objectives. Without any communication constraint (c.f. \ref{figure3c} and \ref{figure4c}), the optimal sensing configuration moves BSs closer to the target layer and reconfigures their layout to minimize the CRLB, where even a modest height adjustment of a single BS combined with a small horizontal displacement yields significant sensing performance gains. Compared to the communication‐optimal deployment, the overall positioning CRLB map is improved by a factor of $4$. In Fig. \ref{figure3c}, it is observed that simply minimizing the distance between BSs and the sensing region tends to place them as close as possible to the targets. However, while a shorter distance increases the scaling factor, it also reduces the angular diversity among the BSs. This causes the row vectors of the Jacobian matrix to become nearly linearly dependent, drastically increasing the condition number of the Fisher information matrix.
To balance the angular spread and measurement accuracy, BSs should be deployed at an appropriate distance from the sensing region, which is similar to the results in Fig. \ref{figure2b}. 

As shown in Fig. \ref{figure4}, when BSs are placed to maximize the communication rate, achieving 6.5 bps/Hz without any regard to sensing, the probability $\Pr(\mathrm{CRLB}\le\Gamma)$ remains low for all $\Gamma$ values up to 0.3. By introducing a moderate rate constraint, i.e., reducing the peak rate to 5 bps/Hz (around 23\% decrease), we observe an approximately 20\% improvement in the sensing coverage overall. Moreover, if the objective is switched to minimizing the CRLB without any communication requirement, the sensing coverage can increase by up to threefold at \(\Gamma=0.1\), demonstrating the significant trade-off between communication throughput and sensing performance under the consideration of BS deployment.

\begin{figure*}[htbp]
	\centering
	\vspace{0mm}
	\subfigure[CRLB map for optimal communication deployment.]{
		\label{figure9a}
		\includegraphics[width=5.2cm]{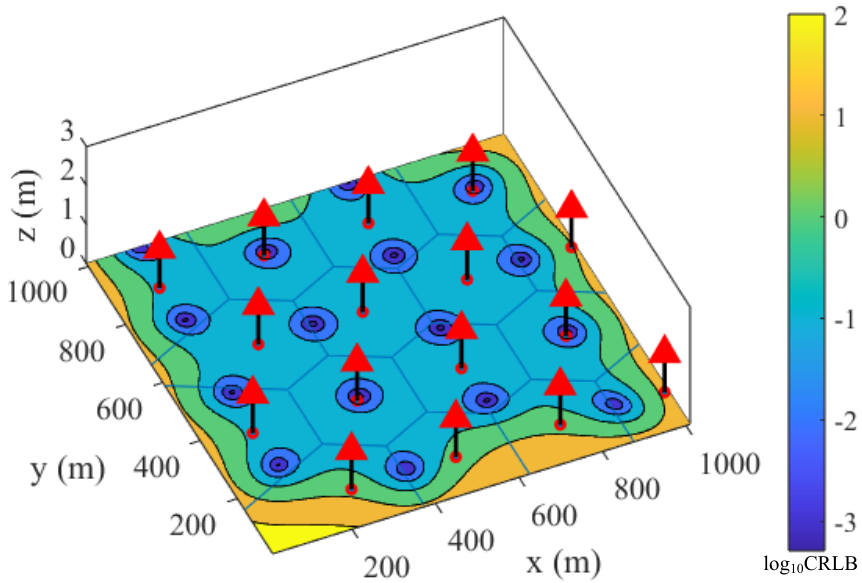}
	}
	\subfigure[CRLB map for optimal sensing deployment.]{
		\label{figure9b}
		\includegraphics[width=5.2cm]{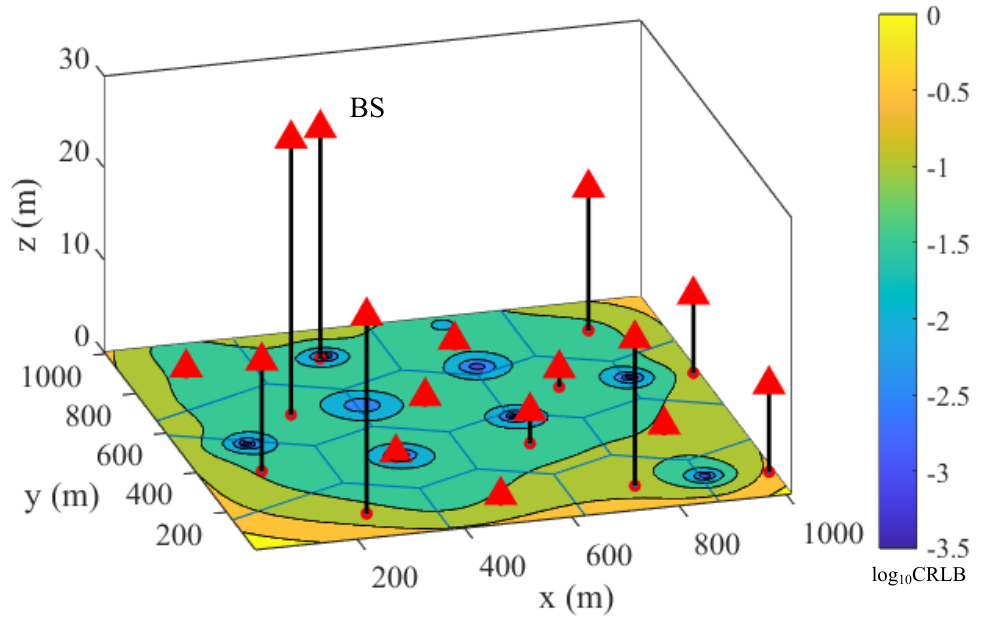}
	}
	\subfigure[Coverage probability comparisons.]{
		\label{figure9c}
		\includegraphics[width=4.5cm]{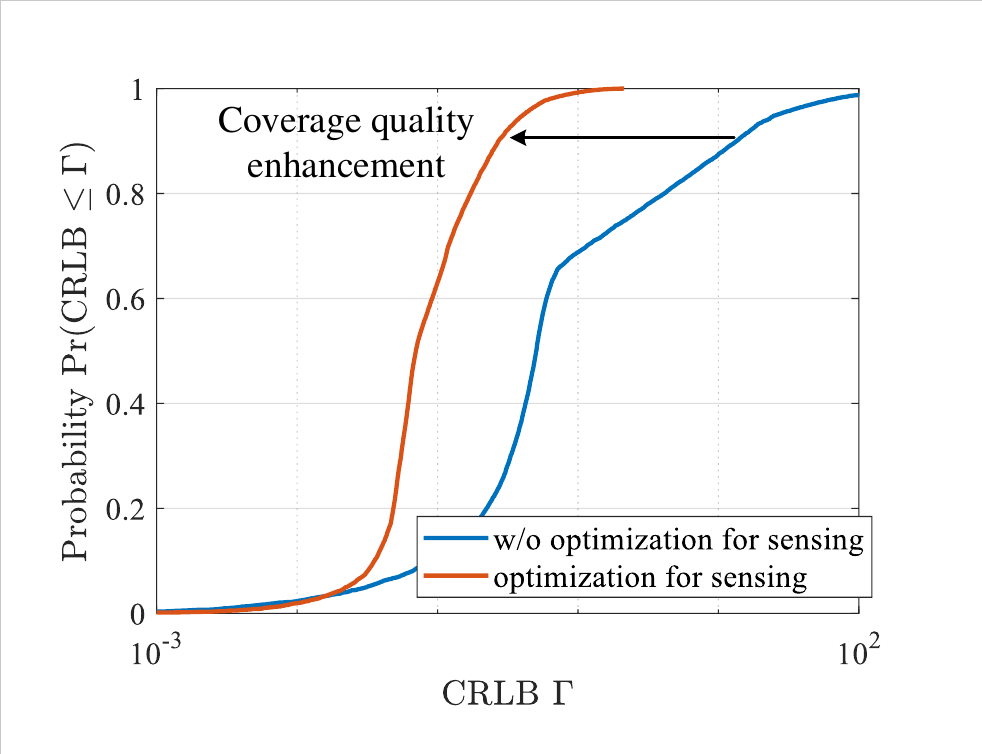}
	}
	\vspace{0mm}
	\caption{Regular horizontal deployment with height optimization.}
	\label{figure9ab}
\end{figure*}

\begin{figure*}[htbp]
	\centering
	\vspace{0mm}
	\subfigure[Target distribution PDF.]{
		\label{figure8a}
		\includegraphics[width=4.7cm]{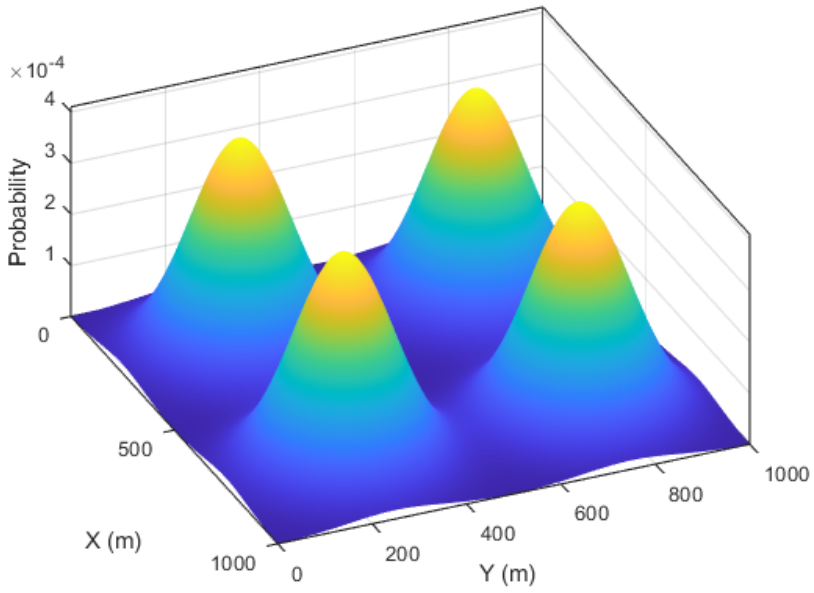}
	}
	\subfigure[Optimal deployment for sensing.]{
		\label{figure8b}
		\includegraphics[width=6.2cm]{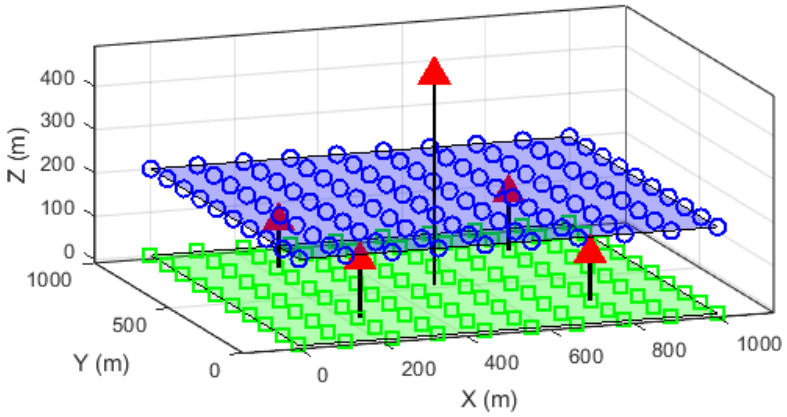}
	}
	\subfigure[Sensing‐optimal CRLB map.]{
		\label{figure8c}
		\includegraphics[width=5cm]{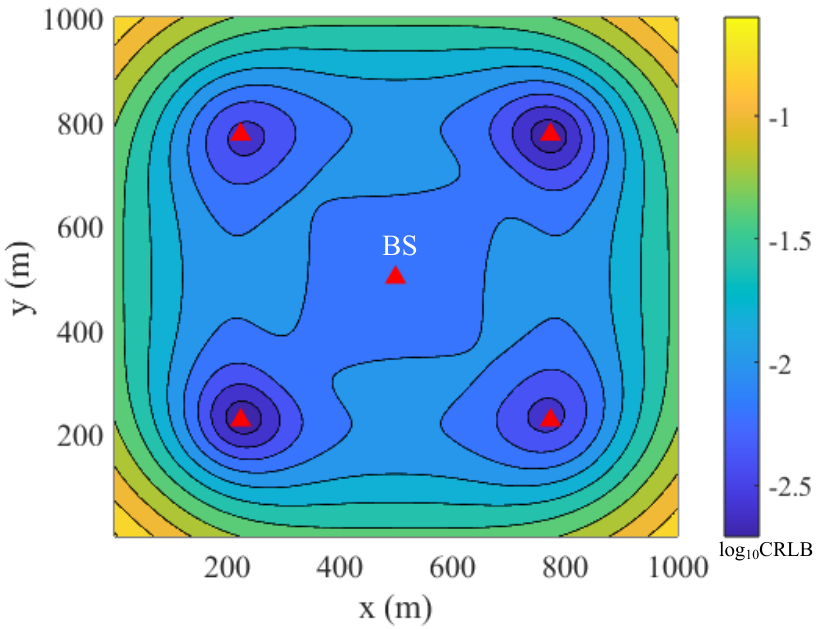}
	}
	\subfigure[Target distribution PDF.]{
		\label{figure8d}
		\includegraphics[width=4.7cm]{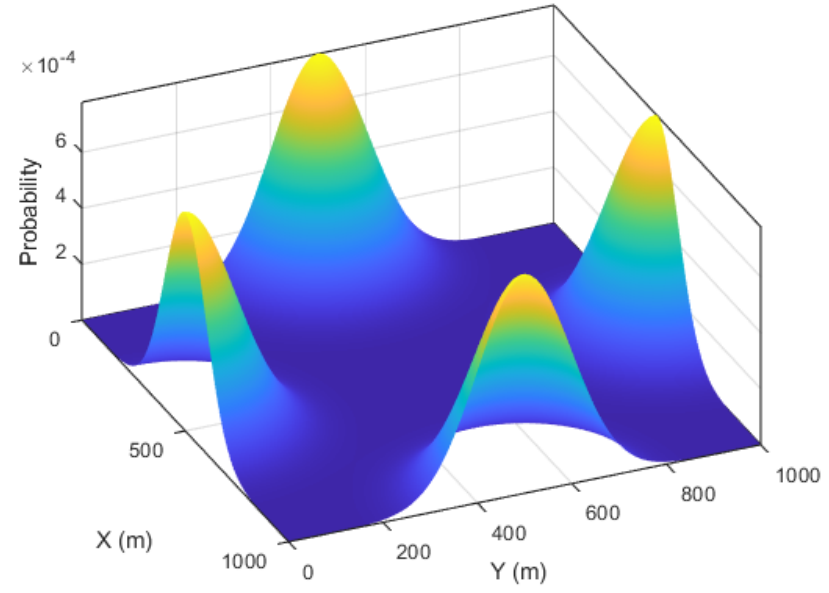}
	}
	\subfigure[Optimal deployment for sensing.]{
		\label{figure8e}
		\includegraphics[width=6.2cm]{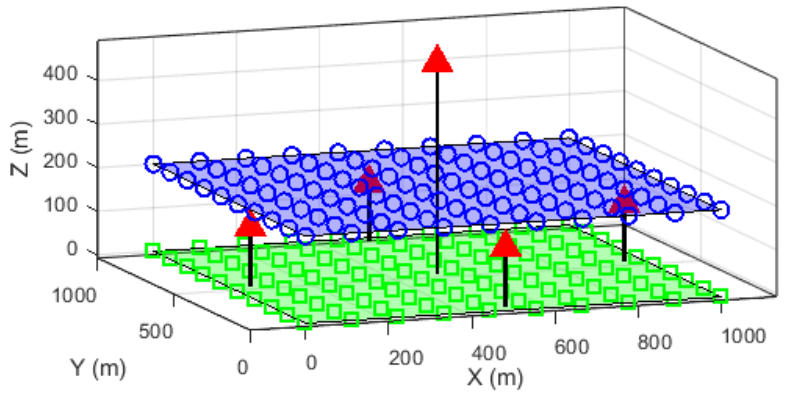}
	}
	\subfigure[Sensing‐optimal CRLB map.]{
		\label{figure8f}
		\includegraphics[width=5cm]{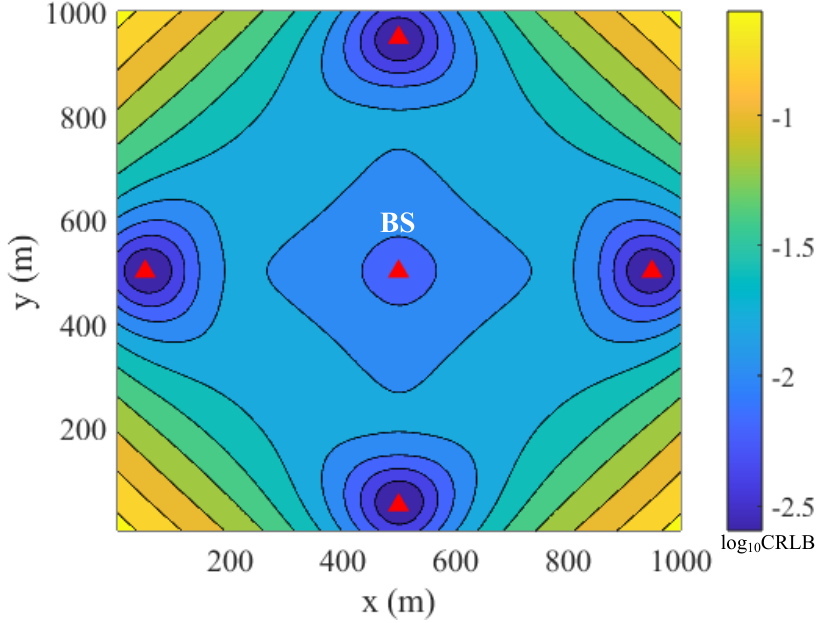}
	}
\vspace{0mm}
	\caption{Comparisons of 3-D BS deployment based on the proposed framework under different target distribution PDF in an area of interest.}
	\label{figure8abs}
\end{figure*}

Fig. \ref{figure9a} considers a conventional cellular‐style layout in which each BS is placed on a regular grid (horizontal coordinates fixed) at a uniform height of 10m. Despite proximity to some BSs, severe sensing‐coverage voids remain, particularly in regions where the geometric intersection angles of the sensing links are acute, a consequence of BSs being deployed relatively close to the target plane. To isolate the impact of antenna elevation, we then retain the same horizontal BS locations but optimize only their heights. As shown in Fig. \ref{figure9b}, this vertical repositioning compresses the CRLB values from the original span of $[10^{-3}, 10^{2}]$ down to $[10^{-5}, 10^{0}]$, reflecting a dramatic gain in localization precision. Finally, by applying a CRLB threshold of $10^{-1}$, Fig. \ref{figure9c} demonstrates that sensing coverage jumps from roughly 10\% to 60\%, illustrating that even modest BS‐height adjustments can yield substantial improvements in network‐level ISAC performance.

In Fig.~\ref{figure8abs}, we compare the BS layouts optimized for various spatial probability density functions (PDF) of sensing targets and communication users.  Let the service area be denoted by \(\mathcal{A}\), and assume that the user and target locations follow the known density functions \(P_u(\mathbf{u})\) and \(P_t(\mathbf{t})\), respectively, which can be estimated \emph{a priori} from historical access logs, traffic heatmaps, or mobility statistics. 
Under these distributions, the optimal BSs naturally cluster around high-density regions to enhance sensing performance.  However, to avoid poor geometric intersection angles that degrade localization accuracy, the deployment must also include at least one elevated pivot BS positioned near the overall centroid of \(\mathcal{A}\). This hybrid strategy both preserves fairness across the service area and significantly improves geometric gain for multi‐static sensing.  Consequently, areas having higher user or target concentration benefit from enhanced sensing coverage, while the pivot BS ensures robust performance in sparsely populated regions.

\begin{figure}[t]
	\centering
	\includegraphics[width=8.4cm]{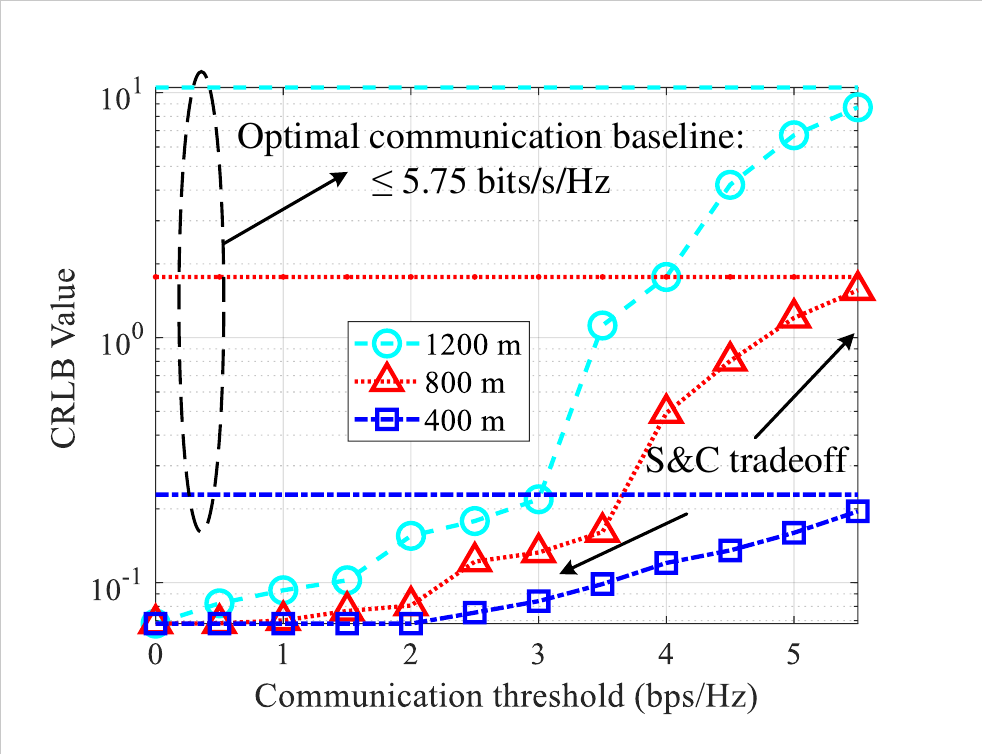}
	\vspace{0mm}
	\caption{S\&C performance trade-offs with optimized BS deployment under different target altitude.}
	\label{figure7}
\end{figure}

As illustrated in Fig. \ref{figure7}, optimizing the BS deployment purely for maximum communication throughput yields a peak rate of 5.75 bps/Hz, but results in uniformly low sensing coverage \(\Pr(\mathrm{CRLB}\le\Gamma)\) across all CRLB thresholds. By imposing a modest throughput constraint, reducing the peak rate to 5 bps/Hz (a 23\% decrease), we achieve roughly a 20\% uplift in sensing coverage across the board. Finally, if the optimization objective is switched to minimizing the CRLB with no communication‐rate requirement, the sensing coverage at \(\Gamma=0.1\) increases by nearly threefold. These findings vividly illustrate the inherent trade‐off between communication throughput and sensing performance when designing BS deployments. Fig. \ref{figure7} illustrates how base-station placement governs the trade-off between communication and sensing performance. As the minimum communication rate requirement is tightened, sensing coverage deteriorates, underscoring the influence of BS geometry on this dual-function trade-off. Remarkably, by perturbing the positions of just a few BSs, at the cost of under a 10\% reduction in peak data rate, sensing performance can be boosted substantially. Compared to the baseline layout optimized solely for maximum throughput, these minor relocations yield a network configuration that simultaneously satisfies both communication and sensing criteria. Moreover, Fig. \ref{figure7} reveals that this trade-off exacerbates at higher target altitudes. Explicitly, as the mean separation between the sensing target region and the communication user region increases, maintaining both objectives simultaneously becomes increasingly difficult.

\section{Conclusions}

A novel deployment framework was proposed for ISAC networks that addresses the fundamental tradeoff between S\&C. Our design ensures both localization coverage and communication performance, guaranteeing high-precision localization and high-throughput data transmission across the entire service region. We analyzed the complexity of cooperative sensing among multiple BSs from a localization-performance coverage perspective, ensuring that the CRLB requirements are met uniformly across the region. By identifying key invariance properties of the optimal deployment, such as shift, rotation, and scaling invariance, we developed a low-complexity algorithmic framework for ISAC network planning.
We formulated the joint sensing-communication optimization as a structured problem amenable to the MM principle, resulting in an MM-based algorithm that reliably converges to high-quality solutions at a low computational cost per iteration. Simulation results demonstrated that the proposed framework significantly enhances sensing coverage, while maintaining communication throughput with minimal performance degradation. This work provides insights for designing practical large-scale ISAC networks. Exploring joint resource scheduling and BS placement in dynamic scenarios offers valuable avenues for future research. Moreover, extending the area CRLB to exploit informative priors, using Bayesian CRLB formulations, merits further investigation.

\vspace{-2mm}
\footnotesize  	
\bibliography{mybibfile}
\bibliographystyle{IEEEtran}

\end{document}